\documentclass[1p,number,a4paper]{elsarticle}
\usepackage{stix}
\usepackage{amsfonts}
\usepackage{amsmath}
\usepackage{amsthm}
\usepackage{mathrsfs}
\usepackage{enumerate}
\usepackage{graphics}
\usepackage{mathtools}
\usepackage{microtype}
\usepackage{xcolor}
\usepackage[czech, english]{babel}
\usepackage{esvect}
\usepackage{complexity}
\usepackage{comment}
\usepackage[many]{tcolorbox}
\usepackage{thm-restate}
\usepackage{mdwtab}
\usepackage[all]{xy}
\usepackage[colorlinks=true,linkcolor={blue},citecolor=blue]{hyperref}
\usepackage[noabbrev,capitalise,nameinlink]{cleveref}
\usepackage[defaultlines=4]{nowidow}

\newcommand{\dr}{d_{\hspace{-1pt}R}}
\newcommand{\Dr}{\Delta_{\hspace{-1pt}R}}
\newcommand{\Nr}{N_{\hspace{-1pt}R}}

\theoremstyle{remark}
\newtheorem{claim}{$\rhd$ Claim}
\newcommand{\cqed}{\ensuremath{\lhd}}
\makeatletter

\tcolorboxenvironment{theorem}{
colframe=cyan,interior hidden, breakable,before skip=10pt,after skip=10pt,	outer arc=0pt,
	arc=0pt,
	colback=white,
	rightrule=0pt,
	toprule=0pt,
	top=0pt,
	right=0pt,
	bottom=0pt,
	bottomrule=0pt,
 }

\tcolorboxenvironment{lemma}{
colframe=cyan,interior hidden, breakable,before skip=10pt,after skip=10pt,	outer arc=0pt,
	arc=0pt,
	colback=white,
	rightrule=0pt,
	toprule=0pt,
	top=0pt,
	right=0pt,
	bottom=0pt,
	bottomrule=0pt,
 }

\tcolorboxenvironment{corollary}{
	colframe=blue,interior hidden, breakable,before skip=10pt,after skip=10pt,	outer arc=0pt,
	arc=0pt,
	colback=white,
	rightrule=0pt,
	toprule=0pt,
	top=0pt,
	right=0pt,
	bottom=0pt,
	bottomrule=0pt,
}

\makeatother
\DeclareTColorBox{answer}{O{orange}O{0cm}}{
	breakable,
	outer arc=0pt,
	arc=0pt,
	colback=white,
	rightrule=0pt,
	toprule=0pt,
	top=0pt,
	right=0pt,
	bottom=0pt,
	bottomrule=0pt,
	colframe=#1,
	enlarge left by=#2,
	width=\linewidth-#2,
}
\DeclareTColorBox{solution}{O{orange}O{0cm}}{
	breakable,
	outer arc=0pt,
	arc=0pt,
	colback=white,
	rightrule=0pt,
	toprule=0pt,
	top=0pt,
	right=0pt,
	bottom=0pt,
	bottomrule=0pt,
	colframe=#1,
	enlarge left by=#2,
	width=\linewidth-#2,
}

\newcommand{\Oof}{\mathcal{O}}

\newcommand{\Cc}{\mathscr{C}}

\newcommand{\Pp}{\mathcal{P}}

\newcommand{\Tt}{\mathcal{T}}

\newcommand{\minor}{\preccurlyeq}

\renewcommand{\phi}{\varphi}

\renewcommand{\epsilon}{\varepsilon}

\usepackage[all]{xy}

\usepackage{pgfplots}
\pgfplotsset{compat = newest}

\newcommand{\ERCagreement}{This paper is part of a project that has received funding from the European Research Council (ERC) under the European Union's Horizon 2020 research and innovation programme (grant agreements No 810115 -- {\sc Dynasnet} and 948057) and from the German
	Research Foundation (DFG) with grant agreement
	No 444419611.\\
			\includegraphics[width=.25\textwidth]{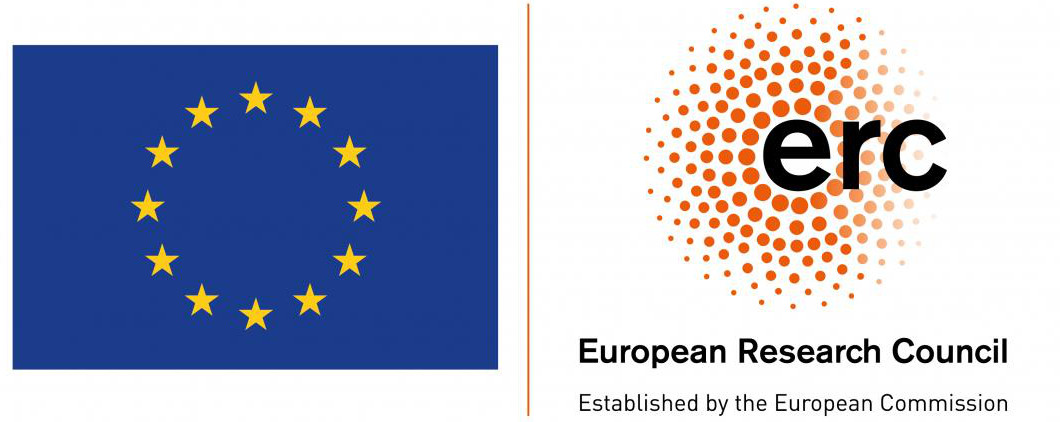}\
			\includegraphics[width=.2\textwidth]{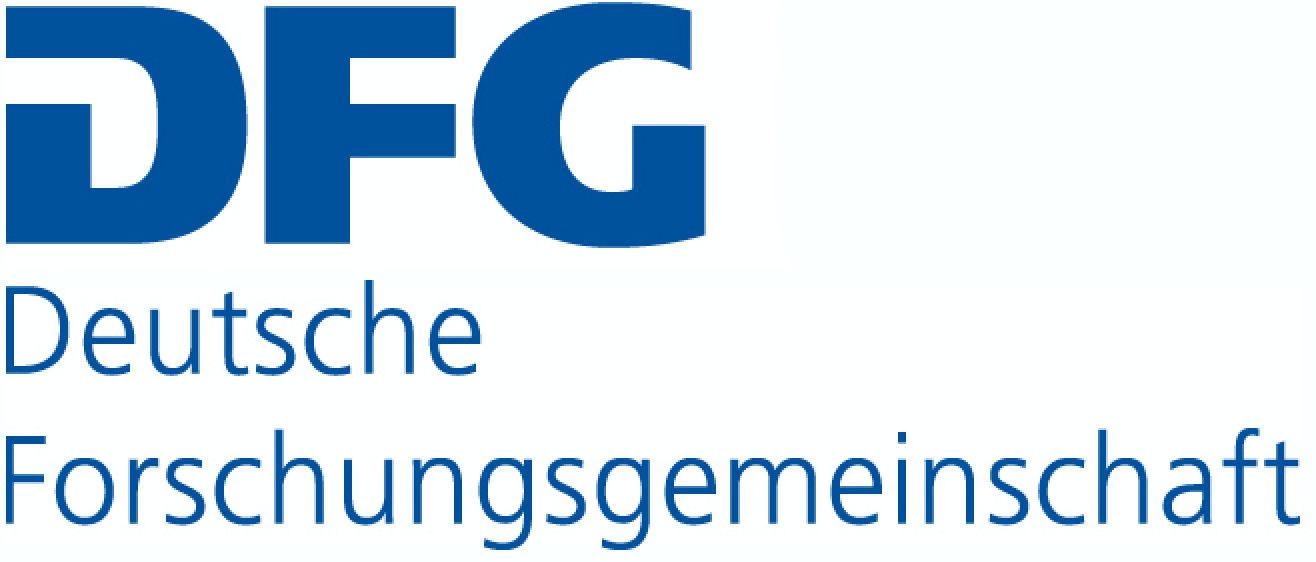}}

\newtheorem{theorem}{Theorem}[section]

\newtheorem{corollary}{Corollary}[section]
\newtheorem{definition}{Definition}[section]
\newtheorem{lemma}{Lemma}[section]

\newtheorem{example}{Example}
\crefname{ext_theorem}{Theorem}{Theorems}
\crefname{corollary}{Corollary}{Corollaries}
\crefname{lemma}{Lemma}{Lemmas}
\crefname{section}{Section}{Sections}

\newcommand{\nc}[1]{\newcommand{#1}}

\nc{\e}{\epsilon}
\nc{\dd}{\delta}
\nc{\dom}{\text{dom}}
\nc{\nn}{\nabla}
\nc{\bn}{\nabla\!\!\!\nabla_1}

\journal{}
\begin{document}
\begin{frontmatter}
% \vspace{-2cm}
\title{Distributed domination on sparse graph classes}
%\tnoteref{This paper subsumes the results of~\cite{kublenz2020distributed} and \cite{}.}}
\tnotetext[ERC]{\ERCagreement}
\author{Ozan Heydt}\address{University of Bremen, Bremen, Germany}\ead{heydt@uni-bremen.de}
\author{Simeon Kublenz}\address{University of Bremen, Bremen, Germany}\ead{kublenz@uni-bremen.de}
\author{Patrice Ossona de Mendez}\address{Centre d'Analyse et de Math\'ematiques Sociales (CNRS, UMR 8557), Paris, France and \\Computer Science Institute of Charles University, Praha, Czech Republic}\ead{pom@ehess.fr}
\author{Sebastian Siebertz}\address{University of Bremen, Bremen, Germany}\ead{siebertz@uni-bremen.de}
\author{Alexandre Vigny}\address{University of Bremen, Bremen, Germany}\ead{vigny@uni-bremen.de}
\begin{keyword}
  Dominating set, Distributed LOCAL algorithms, Bounded expansion graph classes, Planar graphs.\\[2.5mm]

This journal paper subsumes the results of the extended abstracts~\cite{kublenz2020distributed} and \cite{heydt2021local}.
\end{keyword}
% !TEX root = main.tex

\begin{abstract}
  We show that the dominating set problem admits a constant factor
  approximation in a constant number of rounds in the \mbox{LOCAL}
  model of distributed computing on graph classes with bounded
  \mbox{expansion}.  This generalizes a result of Czygrinow et al.\
  for graphs with excluded topological minors to very general classes
  of uniformly sparse graphs.  We demonstrate how our general
  algorithm can be modified and fine-tuned to compute an
  ($11+\e$)\hspace{1pt}\raisebox{0.3pt}{-}\hspace{0.9pt}approximation
  (for any $\epsilon>0)$ of a minimum dominating set on planar
  graphs. This improves on the previously best known approximation
  factor of~52 on planar graphs, which was achieved by an elegant and
  simple algorithm of Lenzen et al.
\end{abstract}

\end{frontmatter}

% !TEX root = main.tex

\section{Introduction}

A dominating set in an undirected and simple graph $G$ is a set
$D\subseteq V(G)$ such that every vertex $v\in V(G)$ either belongs
to $D$ or has a neighbor in $D$.
%The optimization version of the minimum dominating set problem takes as input a graph $G$ and the objective
%is to find a minimum size dominating set of~$G$.
The dominating set problem has many
applications in theory and practice, see e.g.~\cite{du2012connected,sasireka2014applications}, unfortunately
however, already the decision
problem whether a graph admits a dominating set of size $k$
is NP-hard~\cite{karp1972reducibility} and this even holds in
very restricted settings, e.g.\ on planar graphs of maximum degree
$3$~\cite{garey1979computers}.

Consequently, attention
shifted from computing exact solutions to approxi\-mating
near optimal dominating sets. A simple greedy algorithm computes
an $\ln n$ approximation (where $n$ is number of vertices
of the input graph)
of a minimum dominating set \cite{johnson1974approximation,lovasz1975ratio}, and for
general graphs this algorithm is near optimal -- it is NP-hard
to approximate minimum dominating sets within factor
$(1-\epsilon)\ln n$ for every $\epsilon>0$~\cite{dinur2014analytical}.

Therefore, researchers tried to identify restricted
graph classes where better (sequential) approximations are possible.
For example, the problem
admits a PTAS on classes with sub\-exponential expansion~\cite{har2017approximation}. Here, expansion refers to the edge
density of bounded depth minors, which we will define formally
below. Important examples of classes with subexponential
expansion include the class of planar graphs and more generally
classes that exclude some fixed graph as a minor. The dominating
set problem admits a constant factor approximation on classes of
bounded degeneracy (equivalently, of bounded arboricity)~\cite{bansal2017tight,lenzen2010minimum}
and an $\Oof\hspace{1pt}(\ln \gamma)$ approxi\-mation (where~$\gamma$ denotes the size
of a minimum dominating set) on classes of bounded VC-dimension~\cite{bronnimann1995almost,even2005hitting}. In fact, the greedy
algorithm can be modified to yield an $\Oof\hspace{1pt}(\ln \gamma)$
approximation on biclique-free graphs (graphs that exclude some fixed
complete bipartite graph $K_{t,t}$ as a subgraph)~\cite{siebertz2019greedy}
and even a constant factor approximation on
graphs with bounded degeneracy~\cite{jones2017parameterized}.
However, it is unlikely
that polynomial-time constant factor approximations exist even on
$K_{3,3}$-free graphs~\cite{siebertz2019greedy}.
The general goal in this line of research is to identify the broadest
graph classes on which the dominating set problem (or other important
problems that are hard on general graphs) can be approximated
efficiently with a certain guarantee on the approximation factor.
These limits of tractability are often captured by abstract notions, such
as expansion, degeneracy or VC-dimension. % of graph classes.

%\medskip
In this paper we study the distributed time complexity of finding
dominating sets in the classic LOCAL model of distributed computing,
which can be traced back at least to the seminal work of Gallager,
Humblet and Spira~\cite{gallager1983distributed}. In this model, a
distributed system is modeled by an undirected (connected) graph~$G$,
in which every vertex represents a computational entity of the network and every edge represents a bidirectional communication channel. The vertices are equipped with unique identifiers.
In a distributed algorithm, initially, the nodes have no knowledge about
the network graph. They must then communicate and coordinate
their actions by passing messages to one another in order to achieve
a common goal, in our case, to compute a dominating set of the
network graph. The LOCAL model focuses on the aspects of
communication complexity and therefore the main measure for
the efficiency of a distributed algorithm is the number of communication
rounds it needs until it returns its answer.

%\medskip
Kuhn et al.~\cite{KuhnMW16} proved that in~$r$ rounds on an~$n$-vertex graphs of maximum degree
$\Delta$ one can approximate minimum dominating sets only within a factor $\Omega(n^{c/r^{\mathrlap{2}}}/r)$
and~$\Omega(\Delta^{1/(r+1)}/r)$, respectively, where~$c$ is a constant.
This implies that, in general, to achieve a constant approximation ratio,
we need at least $\Omega\hspace{1pt}(\sqrt{\log
    n/\log \log n})$ and~$\Omega\hspace{1pt}(\log \Delta/\log \log \Delta)$ communication rounds, respectively.
Kuhn et al.~\cite{KuhnMW16} also presented a~$(1+\epsilon)\ln \Delta$-approximation that runs in $\Oof\hspace{1pt}(\log(n)/\epsilon)$ rounds for any~$\epsilon>0$,
Barenboim et al.~\cite{barenboim2018fast}
presented a deterministic $\Oof\hspace{1pt}((\log n)^{k-1})$-time algorithm that provides an
$\Oof\hspace{1pt}(n^{1/k})$-approximation, for any integer parameter~$k \ge 2$.
More recently, the combined results of Rozhon, Ghaffari, Kuhn, and Maus~\cite{DBLP:conf/stoc/GhaffariKM17,DBLP:conf/stoc/RozhonG20}
provide an algorithm computing a $(1+\epsilon)$-approximation of the dominating set
in poly$(\log(n)/\epsilon)$ rounds~\cite[Corollary 3.11]{DBLP:conf/stoc/RozhonG20}.

Since by the results of Kuhn et al.~\cite{KuhnMW16} in general graphs it is not
possible to compute a constant factor approximation in a constant number
of rounds,
much effort has been invested to improve the ratio between approximation
factor and number of rounds on special graph classes.
For graphs of degeneracy~$a$ (equivalent to arboricity up to factor $2$),
Lenzen and Wattenhofer~\cite{lenzen2010minimum}
provided an algorithm that achieves a factor~$\Oof\hspace{1pt}(a^2)$ approximation
in randomized time~$\Oof\hspace{1pt}(\log n)$, and a deterministic~$\Oof\hspace{1pt}(a \log
\Delta)$ approximation algorithm
with $\Oof\hspace{1pt}(\log \Delta)$ rounds. Graphs of bounded degeneracy include all graphs that exclude a fixed graph as a (topological) minor and in particular, all planar graphs and any class of bounded genus.

Amiri et al.~\cite{akhoondian2018distributed} provided a deterministic
$\Oof\hspace{1pt}(\log n)$ time constant factor approximation algorithm on
classes of bounded expansion (which extends also to connected
dominating sets).
Czygrinow et al.~\cite{czygrinow2008fast} showed
that for any given~\mbox{$\epsilon>0$}, $(1+\epsilon)$-approximations of a maximum independent
set, a maximum matching, and a minimum dominating set, can be computed in
$\Oof\hspace{1pt}(\log^* n)$ rounds in planar graphs, which is asymptotically optimal~\cite{lenzen2008leveraging}.

Lenzen et al.~\cite{lenzen2013distributed} proved that on planar graphs
a 130\hspace{1pt}\raisebox{0.3pt}{-}\hspace{0.9pt}approximation of a minimum dominating set can be computed in a
constant number of
rounds. A careful analysis of Wawrzyniak~\cite{wawrzyniak2014strengthened}
later showed that the algorithm computes in fact a 52\hspace{1pt}\raisebox{0.3pt}{-}\hspace{0.9pt}approximation.
In terms of lower bounds, Hilke et al.~\cite{hilke2014brief} showed that there is no
deterministic local algorithm (constant-time distributed graph algorithm) that
finds a~$(7-\epsilon)$-approximation of a minimum dominating set on
planar graphs, for any positive constant~$\epsilon$. Better approximation
ratios are known for some special cases, e.g.\ 32 if the planar graph is
triangle-free \mbox{\cite[Theorem 2.1]{alipour2020distributed}}, 18 if the planar graph has girth
five~\cite{alipour2020local} and 5 if the graph is
outerplanar (and this bound is tight)~\cite[Theorem 1]{bonamy2021tight}.
Wawrzyniak~\cite{wawrzyniak2013brief} showed
that message sizes of $\mathcal{O}(\log n)$ suffice to give a
constant factor approximation on planar graphs in a constant number
of rounds.

The constant factor approximations in a constant number of rounds for planar
graphs were gradually extended to classes with bounded genus~\cite{akhoondian2016local,amiri2016brief}, classes with sublogarithmic expansion~\cite{amiri2019distributed} and eventually by Czygrinow et al.~\cite{czygrinow2018distributed} to classes with excluded topological minors.
Again, one of the main goals in this line of research is to find the most general
graph classes on which the dominating set problem admits a constant
factor approximation in a constant number of rounds.

\medskip

\begin{figure}[h!]
\begin{center}
\begin{tikzpicture}

\node (bd-deg) at (11,-2.7) {\scriptsize\textit{bounded degree}};
\node[align=center] (topminor) at (8.5,-1.7) {\scriptsize\textit{excluded}\\[-2mm]\scriptsize\textit{topological}\\[-2mm] \scriptsize\textit{minor}};
\node[align=center] (sublog) at (4.5,-1.7) {\scriptsize\textit{subexponential}\\[-2mm]\scriptsize\textit{expansion}};
\node (bd-exp) at (6.5,-0.8) {\scriptsize\textbf{\textit{bounded expansion}}};
\node (degenerate) at (6.5,0) {\scriptsize\textit{bounded degeneracy}};
\node (biclique-free) at (6.5,0.8) {\scriptsize\textit{biclique-free}};
\node (vc) at (6.5,1.6) {\scriptsize\textit{bounded VC-dimension}};
\node (planar) at (6.5,-4.1) {\scriptsize\textit{planar}};
\node (genus) at (6.5,-3.4) {\scriptsize\textit{bounded genus}};
\node (minor) at (6.5,-2.7) {\scriptsize\textit{excluded minor}};

%%%%%%%%%%%%%%%% Arrows %%%%%%%%%%%%%%%%

\draw[->,>=stealth] (planar) to (genus);
\draw[->,>=stealth] (genus) to (minor);
\draw[->,>=stealth] (minor) to (topminor);
\draw[->,>=stealth] (topminor) to (bd-exp);
\draw[->,>=stealth] (bd-deg) to (topminor);
\draw[->,>=stealth] (minor) to (sublog);
\draw[->,>=stealth] (sublog) to (bd-exp);
%\draw[->,>=stealth] (topminor) to[bend right=10] (bd-exp.east);
\draw[->,>=stealth] (bd-exp) to (degenerate);
\draw[->,>=stealth] (degenerate) to (biclique-free);
\draw[->,>=stealth] (biclique-free) to (vc);

\end{tikzpicture}
\end{center}
\caption{Inclusion diagram of the mentioned graph classes. }
\end{figure}
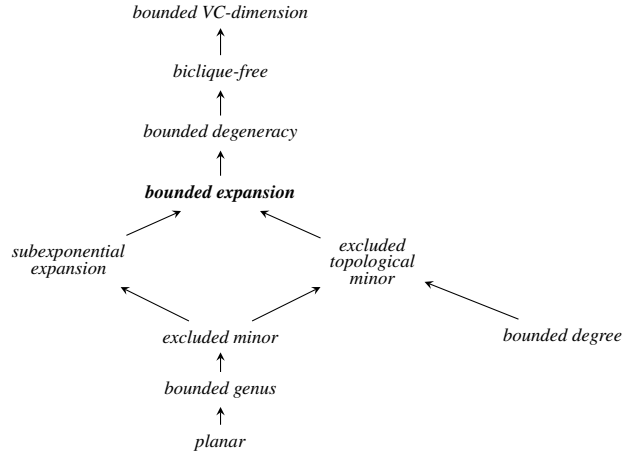\label{fig:classes}

We take a step towards this goal and generalize the result of
Czygrinow et al.~\cite{czygrinow2018distributed} to classes of bounded
expansion. The notion of bounded expansion was introduced
by Ne\v{s}et\v{r}il and Ossona de Mendez~\cite{nevsetvril2008grad} and
offers an abstract definition of uniform sparseness in graphs. It is based on bounding the density of shallow minors. Intuitively, while
a minor is obtained by contracting arbitrary connected subgraphs of a graph
to single vertices, in an $r$-shallow minor we are only allowed to contract
connected subgraphs of radius at most~$r$.

%\medskip
A class of graphs has
bounded expansion if for every radius $r$ the set of all \mbox{$r$-shallow}
minors has edge density bounded by a constant depending only on~$r$.
We write $\nabla_r(G)$ for the maximal edge density of an
$r$-shallow minor of a graph~$G$.
Of course, every class~$\Cc$ that excludes a fixed graph $H$ as
a minor has bounded expansion. For such classes there exists an
absolute constant
$c$ such that for all $G\in\Cc$ and all~$r$
we have $\nabla_r(G)\leq c$.
Special cases are the class of
planar graphs, every class of graphs that can be drawn
with a bounded number of crossings, and every class of graphs
that embeds into a fixed surface.
Every class of intersection graphs of low density objects in low
dimensional Euclidean space has polynomial expansion, that is, the function~$\nabla_r$ is bounded polynomially in $r$ on $\Cc$. Also
every class $\Cc$ that excludes a fixed graph $H$ as
a topological minor has bounded expansion.
Important special cases are classes of
bounded degree and classes of graphs that can be drawn
with a linear number of crossings.
Further examples include
classes of graphs with bounded queue-number, bounded stack-number or bounded
non-repetitive chromatic number. Also, for each constant $d>0$ there is a bounded expansion class $\mathscr R_d$, to which  the Erd\H os-R\'enyi random graphs $G(n,d/n)$ belong
asymptotically almost surely.
See \cite{har2017approximation,nevsetvril2012characterisations} for reference
to  all these examples.

Classes of bounded expansion are more general than
classes excluding a topological minor. However, maybe
not surprisingly, when performing local
computations, it is not properties of minors or topological minors, but
rather of shallow minors that enable the necessary combinatorial arguments
in the algorithms. This observation was already made in the study of the kernelization complexity of dominating set on classes of sparse graphs \cite{DrangeDFKLPPRVS16,eiben2019lossy,EickmeyerGKKPRS17,FabianskiPST19,kreutzer2018polynomial}.
Moreover, bounding the edge density of shallow minors might be needed only up to some depth.
For example, degenerate classes
are those classes where only $\nabla_{\hspace{-1pt}0}(G)$, the edge
density of subgraphs, is bounded, and
these classes are  more general than classes of bounded
expansion.

The algorithm of Czygrinow et al.~\cite{czygrinow2018distributed} for
classes excluding a topological minor is
based on an quite complicated iterative process of choosing dominating
vertices from so called
\emph{pseudo-covers}. Based on the fact that classes excluding a topological minor in particular exclude some complete bipartite graph
$K_{t,t}$ as a subgraph, it is proved that this iterative process terminates
after at most $t$ rounds and
produces a good approximation of a minimum dominating set.

\section{Our contribution}
Our contribution is threefold:
First, we simplify the
arguments used by Czygrinow et al.\ and give a  more accessible
description of their algorithm.
Second, we identify the boundedness of~$\nabla_{\hspace{-1pt}1}(G)$ as the key property that makes the algorithm
work. Classes with only this restriction are less general than degenerate classes, but
 more general than bounded expansion classes.
 We generalize the algorithm to these general classes
 and prove that the pseudo-covering method cannot
be extended further, e.g.\ to classes of bounded degeneracy.
Last, we optimize the bounds that arise in
the algorithm in terms of several parameters.
Czygrinow et al.\ explicitly stated that they did not aim to
optimize any constants, and as presented, the constants in their
construction are enormous.
%
% of $\nabla_{\hspace{-1pt}1}(G)$ (in fact,
%to optimize further, we will later work with the edge density of
%bipartite shallow minors, which we will simply denote by $\nn(G)$).
%\textcolor{red}{We demonstrate
%the improvements in particular in the case of planar graphs. We show
%that the algorithm provides the best known approximation on
%planar graphs. The approximation ratio seems to be
%$2+3\cdot 7+9=33$. The best currently known bounds are $52$.
%This will be bigger because we need the neighborhoods to be
%larger than $k^t(2t+tK)$... Check this.}
%\sebi{If we have not enough time we drop the planar case.}
Even though the constants in our analysis
are still large, they are by magnitudes smaller than those in the
original presentation. The following is our first main theorem.

\begin{theorem}\label{thm:main-general}
Let $\nabla_1>0$ be an integer.
There exists a LOCAL algorithm that computes  in a constant number of rounds, for any input graph $G$ with $\nabla_{\hspace{-1pt}1}(G)\le \nabla_1$, a dominating set
of size~$\mathcal{O}\hspace{1pt}(\gamma(G))$, where
$\gamma(G)$ denotes the size of a minimum dominating set of $G$.
\end{theorem}

Note that the algorithm depends on the constant $\nabla_1$.
The reason for this is that we cannot locally compute or approximate
$\nabla_1(G)$
in a constant number of rounds. This is also
the case for the algorithm of Czygrinow et al., which works with the
assumption that the inputs exclude a complete graph $K_t$ with~$t$~vertices
as a topological minor, a property that can also not be verified locally.
Furthermore, the number of rounds depends on $\nabla_1$.

The algorithm is actually tuned using more parameters, like upper bounds on $\nabla_0(G)$ and
on integers $s$ and $t$ such that the complete bipartite graph $K_{s,t}$
is not subgraph of $G$, in order to  improve the approximation ratio of the algorithm.
However, all these parameters can be upper bounded in terms of $\nabla_1$.
When these parameters are given, the algorithm computes a
$(2(\nabla_0+1)((2\nabla_1)^{4s\nabla_1}+2)\gamma$ approximation
in a number of rounds that depends on $\nabla_1$ and $t$.

%\medskip
%As mentioned above, the method of pseudo-covers cannot be extended
%beyond classes where~$\nabla_{\hspace{-1pt}1}$ is uniformly bounded
%by a constant. Hence, one of the main
%open questions that remains in this line of research is whether we can
%compute constant factor approximations of minimum dominating sets
%in a constant number of rounds in classes of bounded degeneracy, that
%is, classes with uniformly bounded $\nabla_{\hspace{-1pt}0}$.
%Another remaining problem is to further optimize the constants
%in the constructions.

\medskip

Then, we modify and fine-tune the algorithm for graphs
excluding $K_{3,t}$ as a subgraph (and having~$\nabla_1$ bounded).
Important examples of graphs with this property are graphs that
can be embedded into a fixed surface of bounded genus. We prove
the following theorem.

\begin{theorem}\label{thm:K3t-free-total}
Let $\nabla_1>0$ and $t\geq 3$ be integers.
There exists a LOCAL algorithm and a function $C$ that for every
$K_{3,t}$-free graph $G$ with $\nabla_1(G)\leq \nabla_1$ and every
$\epsilon>0$, computes in~$C(\epsilon)$ rounds a dominating
set of size at most $(6\nabla_1+3)\gamma$.
\end{theorem}

Since planar graphs satisfy $\nabla_1\leq 3$ and exclude $K_{3,3}$ as
a subgraph, from \cref{thm:K3t-free-total} we can derive an approximation
factor of $21$ for planar graphs. A more careful analysis leads to the
following theorem.

\begin{theorem}\label{thm:planar}
  There exists a LOCAL algorithm and a function $C$ that for every planar
  input graph~$G$ and $\epsilon>0$, computes in $C(\epsilon)$ rounds a dominating set
  of size at most \mbox{$(11+\e)\cdot\gamma(G)$}.
\end{theorem}

We further analyze our algorithm on restricted classes of planar graphs and
improve the upper bounds in several cases (see \cref{tab:approx_factor}).
%In particular, we tighten the gap between
%the best-known lower bound of~$7$ and the best-known upper
%bound of~$52$ for the approximation ratio on planar graphs,  by showing that the new algorithm
%computes an ($11+\e$)\hspace{1pt}\raisebox{0.3pt}{-}\hspace{0.9pt}approximation on planar graphs.

\begin{table}[h!t]
	\begin{tabular*}{\textwidth}{|c|c|c|c|}
		\hlx{hv}
Graph class&Lower bound&Previous upper bound&Our upper bound\\
\hlx{vhhv}
Planar graphs&\hfill$7-\e$\hfill\cite{hilke2014brief}&\hfill52\hfill\cite{wawrzyniak2014strengthened}&$11+\e$\\
Triangle-free planar graphs&&\hfill$32$\hfill\phantom{5}\cite{alipour2020distributed}&$8+\e$\\
Bipartite planar graphs&&&$7+\e$\\
Outerplanar graphs&\hfill$5$\hfill\phantom{3}\cite{bonamy2021tight}&\hfill$5$\hfill\phantom{5}\cite{bonamy2021tight}&$(8+\e)$\\
Girth $\ge 5$ planar graphs&&\hfill$18$\hfill\phantom{5}\cite{alipour2020local}&$7$\\
\hlx{vh}
	\end{tabular*}
	\caption{Approximation factors for a LOCAL approximation of $\gamma(G)$ is a constant number of rounds. Our algorithm improves the approximation factors in all these cases, except for the class of outerplanar graphs.
	}
	\label{tab:approx_factor}
\end{table}

%, and prove that
%it computes approximations within a factor of ($8+\e$) for triangle-free planar graphs,
%of ($7+\e$) for bipartite planar graphs, of ($8+\e$) for outerplanar graphs, and
%of $7$ for planar graphs of girth~$5$
%(\cref{thm:bip,thm:tri,thm:girth,thm:outer}).
%% We then analyze our algorithm on bipartite planar graphs,
%% triangle-free planar graphs,
%% planar graphs of girth $5$ and outerplanar graphs and prove that
%% it computes a 12, 14, 7 and 12 approximation, respectively (\cref{thm:bip,thm:tri,thm:girth,thm:outer}).
%This
%improves the currently best known approximation ratios of~32
%and~18 for triangle-free planar graphs and planar graphs of girth~$5$,
%respectively, while our algorithm falls short of achieving
%the optimal approximation ratio of~5 on outerplanar graphs.

\medskip
Before we go into the technical details, let us give an overview of the
algorithm. The algorithm works in three phases. Each phase ($i\in \{1,2,3\}$) computes a small set $D_i$ that is added to the output dominating set
(where by a  small set we mean a set whose size is linear in $\gamma$).

The first phase is a preprocessing phase, which was similarly employed in the
algorithm of Lenzen et al.~\cite{lenzen2013distributed}. In a key lemma,
Lenzen et al.\ proved that for planar graphs
there are only few vertices whose open neighborhood
cannot be dominated by at most six vertices. This lemma generalizes
to graphs $G$ where $\nabla_{\hspace{-1pt}1}(G)$ is bounded as
shown in~\cite{amiri2019distributed} (where six is replaced by a different constant depending
on $\nabla_{\hspace{-1pt}1}(G)$). We improve this general lemma and
derive in particular that in the case of planar graphs there are only few
vertices whose open neighborhood cannot be
dominated by \emph{three} other vertices.
We pick these few vertices as the set $D_1$,
remove them from~$G$ and mark all their neighbors as dominated.
Hence, after the first phase the open neighborhoods of all
remaining vertices can be dominated by a constant number of
other vertices.

In the second phase, %in parallel for every vertex $v=v_1$
we compute concurrently for each vertex $v$ all the so-called
\emph{domination sequences} $v_1,\ldots, v_s$ starting at $v$ (see
 \cref{def:dom-sequence} for a formal definition). The analysis of this phase is based
on the construction of
\emph{pseudo-covers} as in the work of Czygrinow et al.~\cite{czygrinow2018distributed} and in the approach of greedy domination
in biclique-free graphs~\cite{siebertz2019greedy}. The domination
sequences intuitively
provide a tool to carry out a fine-grained analysis of the vertices that
can potentially dominate the remaining non-dominated neighborhoods.
All the vertices $v_s$ are gathered in the set~$D_2$ and are removed from $G$, with all their neighbors marked as dominated. For $K_{3,t}$-free graphs, we slightly modify the
algorithm and provide an even finer analysis.

Call  the number  $\dr(v)$ of non-dominated neighbors of a vertex $v$
the \emph{residual degree} of $v$. We prove that
after the second phase we are left with a graph where every vertex
has residual degree at most~$\Dr$ for a constant $\Dr$. In particular, every vertex
from a minimum dominating set of size $\gamma$ can dominate at most $\Dr+1$
non-dominated vertices (each
vertex dominates its neighbors and itself) and we conclude that the set $R$
of non-dominated vertices has size bounded by~$(\Dr+1)\gamma$ . Hence, we could at this
point pick all non-dominated vertices to add at most $(\Dr+1)\gamma$ vertices
and conclude. Instead, we study two different ways to proceed with a third phase.

Our first option for the third phase is to apply an LP-approximation based on
results of Bansal and Umboh~\cite{bansal2017tight}, who showed
that a very simple selection procedure leads to a constant factor
approximation when the solution to the dominating set linear
program (LP) is given. As shown by Kuhn et al.~\cite{kuhn2006price} we can approximate such a solution in a
constant number of rounds when the maximum degree~$\Delta$
of the graph is bounded.
To apply these results, we have to overcome two obstacles. First,
note that even though we have established that the maximum residual
degree is bounded by a constant~$\Dr$, we may still have unbounded maximum
degree~$\Delta$. We overcome this problem by keeping only a few
representative potential dominators around the set $R$ of non-dominated
vertices. By a simple
density argument, there can be only very few high degree vertices left
that we simply select into the dominating set. As a result, we are left with
a graph where $\Delta$ is bounded by a constant. The second obstacle,
which is easily overcome, is that we do not need to dominate the
whole remaining graph but only the set $R$. This requires a small
adaptation of the LP-formulation of the problem and a proof that the
algorithm of Bansal and Umboh still works for this slightly different setting.
In total, in this version of the third phase of the algorithm, we add at most $\Oof(\gamma)$
vertices.
%and in particular, if $G$ is planar at most $(7+\e)\gamma$ vertices to
%the dominating set. This in total leads to a constant factor approximation and
%in the planar case to an $(11+\e)$\hspace{1pt}-\hspace{1pt}approximation.  %(\cref{thm:planar}).

Our second option for the third phase is to design a  distributed version of the
classical greedy algorithm.
We proceed in a greedy manner in $d$ rounds, as follows (where
$d$ is a bound on the maximum residual degree $\Dr$ of the graph after phase 2).
In the first round, if a non-dominated
vertex has a neighbor of residual degree~$d$, it elects one such neighbor into
the dominating set (or if it has residual degree~$d$ itself, it may choose itself).
The neighbors of the chosen elements are
marked as dominated and the residual degrees are updated. Note that
all non-dominated neighbors of a vertex of residual degree~$d$
in this round select a dominator. Hence, the residual degrees
of all vertices of residual degree~$d$ are decreased to~$0$ and, after
this round, there are no vertices of residual degree $d$ left.
In the second round, if a non-dominated vertex has a neighbor
of residual degree~$d-1$, it elects
one such vertex into the dominating set, and so on, until after $d$ rounds
in the final round every vertex selects a dominator. Unlike in the general
case, where
nodes cannot learn the current maximum residual degree in a constant
number of rounds, by establishing
an upper bound on the maximum residual degree and proceeding in exactly
this number of rounds, we ensure that we iteratively exactly selects the
vertices of maximum residual degree. For the case of planar graphs,
we prove that $\Dr\leq 30$. It remains to analyze the performance
of this algorithm.

A simple density argument shows
that there cannot be too many vertices of degree $i\geq 2\nabla_{\hspace{-1pt}0}(G)$ ($i\geq 6$ in a planar graph). At a first glance it seems that the algorithm would perform worst
when in every of the $\dr$ rounds it would pick as many vertices as possible,
as the constructed dominating set would grow as much as possible. However,
this is not the case, as picking many high degree vertices at the same time makes
the largest progress towards dominating the whole graph. It turns
out that there is a delicate balance between the vertices that we pick
in round $i$ and the remaining non-dominated vertices that leads
to the worst case.
%We formulate these conditions as a
%linear program and solve the linear program to obtain our claimed bounds.
For planar graphs in total, this
leads to a 20\hspace{1pt}\raisebox{0.3pt}{-}\hspace{0.9pt}approximation.
While the greedy algorithm falls short of achieving the best approximation
factor, it is much simpler than the LP-based approach, and interesting to
analyze in its own right.

%\bigskip
%\noindent\hrulefill
%\vspace{0mm}
%
%\textit{Acknowledgements.} We thank the anonymous
%referee of the conference submission \cite{heydt2021local}  for pointing
%out that the third phase of our algorithm can be
%improved by the use of LP-based techniques when the maximum
%degree is bounded.

%\medskip
%\textcolor{red}{It may be interesting to discuss also parameterized
%complexity and then provide an fpt algorithm running in $f(\gamma)$
%rounds in the Congest model. In a $K_{t,t}$-free graph find a node
%of locally maximum degree (break ties by ids). Mark its neighbors,
%which helps to find the vertex with largest intersection and iterate.}

% !TEX root = main.tex

\section{Preliminaries}

In this paper we study the distributed time complexity of finding
dominating sets in undirected and simple graphs in the classical
LOCAL model of distributed computing.
We assume familiarity with this model and refer to the survey~\cite{suomela2013survey} for an extensive
background on distributed computing and on the LOCAL model.

We use standard notation from graph theory and refer to the
textbook~\cite{diestel} for extensive background.  All graphs in this
paper are undirected and simple. We write $V(G)$ for the vertex set of
a graph $G$ and~$E(G)$ for its edge set. The \emph{girth} of a graph
$G$ is the length of a shortest cycle in~$G$. A graph is 
\emph{triangle-free} if it does not contain a triangle (that is, a
cycle of length three) as a subgraph. Equivalently, a triangle-free
graph is a graph with  girth at least four.

A graph is \emph{bipartite} if its vertex set can be partitioned into
two parts such that all its edges are incident with two vertices from
different parts. We write $K_{s,t}$ for the complete bipartite
graph with parts of size $s$ and $t$, respectively. A set $A$ is
\emph{independent} if no two vertices $u,v\in A$ are
adjacent. We write $\alpha(G)$ for the size of the largest
independent set in a graph~$G$. The \emph{Hall ratio} $\rho(G)$ of
$G$ is defined as $\max\bigl\{|V(H)|/\alpha(H)\mid H\subseteq G\bigr\}$.
Hence, every subgraph $H$ of $G$ contains an independent set of
size at least $|V(H)|/\rho(G)$.

%More generally, the \emph{chromatic number}~$\chi(G)$
%of a graph $G$ is the minimum number $k$ such that the vertices of $G$
%can be partitioned into $k$ parts such that every edge connects
%two vertices from different parts. Hence, the bipartite graphs
%are exactly the graphs with chromatic number two.
%Every graph $G$ contains an independent set of size at
%least $\lceil|V(G)|/\chi(G)\rceil$.

A graph~$H$ is a \emph{minor} of a
graph~$G$, written~$H\minor G$, if there is a set
\mbox{$\{G_v :v\in V(H)\}$} of  vertex disjoint and connected
subgraphs $G_v\subseteq G$ such that if~$\{u,v\}\in E(H)$, then there
is an edge between a vertex of~$G_u$ and a vertex of~$G_v$. We 
say that 
$V(G_v)$ is the \emph{branch set} of $v$ and say that it is
\emph{contracted} to the vertex~$v$. 
The set~\mbox{$\{G_v :v\in V(H)\}$} is a \emph{minor model} of $H$.
The  \emph{depth} of a minor model is the maximum radius of one of its branch sets.
 We call $H$ a \emph{$1$-shallow
  minor}, written~$H\minor_1 G$, if $H\minor G$ and there is a minor
model \mbox{$\{G_v :v\in V(H)\}$} with depth at most $1$ witnessing this.
In other
words, $H\minor_1 G$ if $H$ is obtained from $G$ by deleting some
vertices and edges and then contracting a set of pairwise disjoint
stars. We refer to~\cite{nevsetvril2012sparsity} for an in-depth study
of the theory of sparsity based on shallow minors.

A graph is \emph{planar} if it can be embedded in the plane, that is,
it can be drawn on the plane in such a way that its edges intersect
only at their endpoints. By the famous theorem of Wagner, planar
graphs can be characterized as those graphs that exclude the complete
graph $K_5$ on five vertices and the complete bipartite $K_{3,3}$ with
parts of size three as a minor. In particular, a minor of a planar graph
is again planar. 

A graph is \emph{outerplanar} if it has an embedding in the plane such
that all vertices belong to the unbounded face of the embedding.
Equivalently, a graph is outerplanar if it does not contain the
complete graph $K_4$ on four vertices and the complete bipartite graph
$K_{2,3}$ with parts of size~$2$ and~$3$, respectively, as a minor. 
Again, a minor of an outerplanar graph is again outerplanar. 

By Euler's formula, planar graphs are sparse: every planar $n$-vertex
graph ($n\geq 3$) has at most $3n-6$ edges (and a graph with at most
two vertices has at most one edge).
The ratio $|E(G)|/|V(G)|$ is
called the \emph{edge density} of $G$.
In particular, every planar
graph $G$ has edge density strictly smaller than three.
We define
\begin{align*}
\nabla_0(G)&=\max\Bigl\{\ \frac{|E(H)|}{|V(H)|}\ \Bigr|\  H\subseteq G\Bigr\},\\
\nabla_1(G)&=\max\Bigl\{\ \frac{|E(H)|}{|V(H)|}\ \Bigr| \  H\minor_1 G\Bigr\},\\
\nabla_1^B(G)&=\max\Bigl\{\ \frac{|E(H)|}{|V(H)|}\ \Bigr| \  H\minor_1 G, H\text{ bipartite}\Bigr\}.
\end{align*}

  It is immediate 
that $\nabla_0(G) \le \nabla_1^B(G)\le\nabla_1(G)$. Note that every graph $G$ (and each of its subgraphs)
contains a vertex with degree at most $2\nabla_0(G)$. By iteratively
removing a minimum degree vertex and its neighbors, we can  find a 
large independent set, as stated in the next lemma. The bounds for graphs
on surfaces are well known.

% Every
% graph $G$ has a vertex of degree at most
% $2\nabla_0(G)$ and is $2\nabla_0(G)+1$-colorable.
%
% \begin{lemma}\label{lem:densities}
%   Let $G$ be a planar graph. Then the edge density of $G$ is strictly
%   smaller than $3$ and $\chi(G)\leq 4$. Furthermore,
%
%   \vspace{-2mm}
%   \begin{enumerate}
%   \item if $G$ is bipartite, then the edge density of $G$ is strictly
%     smaller than $2$ and $\chi(G)\leq 2$,\\[-6mm]
%   \item if $G$ is triangle-free or outerplanar, then the edge density
%     of $G$ is strictly smaller than $2$ and $\chi(G)\leq 3$.
%     % and\smallskip
%     % \item if $G$ has girth at least $5$, then the edge density of
%     %   $G$ is strictly smaller than $\frac{5}{3}$ and
%     %   $\chi(G)\leq 3$.
%   \end{enumerate}
% \end{lemma}
\begin{lemma}\label{lem:bounds}
For all graphs $G$ we have $\rho(G)\leq 2\nabla_0(G)+1$. Furthermore,
  \begin{enumerate}
    \item if $G$ is planar, then $\nabla_0(G)<3$, $\nabla_1^B(G)<2$ and $\rho(G)\le 4$;
    \item if $G$ is outerplanar, or planar and triangle free, then $\nabla_0(G)<2$ and $\rho(G)\le 3$;
    \item if $G$ is planar and bipartite, then $\nabla_0(G)<2$ and $\rho(G)\le 2$.
  \end{enumerate}
\end{lemma}

For a graph $G$ and $v\in V(G)$ we write
$N(v)=\{u\colon\{u,v\}\in E(G)\}$ for the \emph{open neighborhood} of $v$
and $N[v]=N(v)\cup\{v\}$ for the \emph{closed neighborhood}
of~$v$. For a set $A\subseteq V(G)$ let $N[A]=\bigcup_{v\in A}N[v]$.
% We write $N_r[v]$ for the set of vertices at distance at most $r$
% from a vertex $v$.
A~\emph{dominating set} in a graph~$G$ is a set $D\subseteq V(G)$ such
that $N[D]=V(G)$. We write~$\gamma(G)$ for the size of a minimum
dominating set of $G$. For a set $R\subseteq V(G)$ we say that a set
$Z\subseteq V(G)$ \emph{dominates} or \emph{covers}
$R$ if $R\subseteq N[Z]$. For $v\in V(G)$ we let 
$\Nr(v)=N(v)\cap R$ and $\dr(v)=|\Nr(v)|$.

%\smallskip
% In the following we mark important definitions and assumptions about
% our input graph in gray boxes and steps of the algorithm in red boxes.

An \emph{orientation} of a graph $G$ is a directed graph $\vec{G}$ that
 has exactly one of the arcs $(u,v)$ and
$(v,u)$
for each edge $\{u,v\}\in E(G)$. The \emph{out-degree} $d^+(v)$ of a vertex $v$ in  $\vec{G}$
is the number of arcs leaving~$v$. The following
lemma is implicit in the work of Hakimi~\cite{sup_orie}, see \mbox{also~\cite[Proposition 3.3]{nevsetvril2012sparsity}.}

\begin{lemma}\label{lem:orientations}
  Every graph $G$ has an
  orientation with maximum out-degree $\nabla_0$.
\end{lemma}
%\begin{proof}
%  We start with an arbitrary orientation $\vec{G}$. Assume there is a vertex $v$
%  with out-degree greater than~$d$. Consider the subgraph $\vec{H}$
%  induced by all vertices that are reachable from $v$. The underlying undirected
%  graph $H$ has $\sum_{v\in V(H)}d^+(v)$ many edges. As $v$ has
%  out-degree strictly larger than $d$ we conclude that there must be one
%  vertex $w$ in $\vec{H}$ that has out-degree strictly smaller than $d$. By definition
%  of $\vec{H}$ there exists a directed path from $v$ to $w$. By reversing the
%  edges on that path we decrease the out-degree of $v$ by one, increase the
%  out-degree of $w$ by one and leave the out-degrees of all other vertices of
%  the path unchanged. By iterating this procedure we successively improve
%  the orientation until no vertex of out-degree greater than $d$ remains.
%\end{proof}

%As every subgraph of a (triangle-free or outerplanar) planar graph
%is again a (triangle-free or outerplanar) planar graph we
%have the following corollary.
We immediately deduce the next corollary.

\begin{corollary}\label{cor:planar-orientations}
  Let $G$ be a planar graph. Then
  \begin{enumerate}
    \item $G$ has an orientation with maximum out-degree $3$.%\smallskip
    \item If $G$ is triangle-free or outerplanar, then
    $G$ has an orientation with maximum out-degree~$2$.
  \end{enumerate}
\end{corollary}

\begin{tcolorbox}
  \begin{tabular}{p{.45\textwidth}p{.5\textwidth}}
  	  In the following we fix&and (possibly defined in terms of $\nabla_1$)\\
  	  \hlx{vv}
  	\begin{itemize}
  		\item the input graph $G$,
  		\item a minimum dominating set~$D$,
  		\item $\gamma \coloneqq |D|$,
  		\item the parameter $\nabla_1\ge \nabla_1(G)$,
  	\end{itemize}
  	&
  	\begin{itemize}
  		\item the parameter $\nabla_0\in[\nabla_0(G),\nabla_1]$,
  		\item the parameter $\nn\in (\nabla_1^B(G),\nabla_1+1]$,
  		\item the parameters $s\leq t$ with $K_{s,t}\not\subseteq G$,
  		%\item the parameter $g\ge g(G)$,
  	\end{itemize}\\
  \hlx{vv}
  \end{tabular}\\
  where $D$ and $\gamma$ are used only to analyze the performance of the algorithm, and where all the  parameters are integers.
%  \begin{itemize}
%  \item a graph $G$,
%  \item a minimum dominating set~$D$ of $G$,
%  \item $\gamma \coloneqq |D|$,
%  \item $\nabla_0\ge \nabla_0(G)$,
%  \item $\nabla_1\ge \nabla_1(G)$,
%  \item $\nn > \nabla_1^B(G) $,
%  \item integers $s,t$ such that $K_{s,t}\not\subseteq G$, and
%  \item the genus $g$ of $G$.
%  \end{itemize}
\end{tcolorbox}

Above, $[a,b]$ denotes the closed real interval containing all 
$x$ with $a\leq x\leq b$ and $(a,b]$ denotes the half-open 
interval containing all $x$ with $a<x\leq b$. We can choose
$s=t=2\lfloor\nabla_0(G)\rfloor+1$.

%The algorithm can be called with a set of parameters bounding
%$\nabla_0$, $\nabla_1, s,t$ and $g$ to give the best approximation
%ratio. It does not require all of these parameters as input though,
%it is also sufficient to give as input only a bound on $\nabla_1$, as $\nabla_1$
%functionally bounds all other parameters as follows.

% !TEX root = main-alex.tex

\section{Phase 1: Preprocessing}\label{sec:step1}
% \vspace{-1mm}
\subsection{Small neighborhoods dominators}
% \vspace{-1mm}
As outlined in the introduction, our algorithm works in three phases.
In phase~$i$ for $1\leq i\leq 3$ we select a partial dominating set
$D_i$ and estimate its size in comparison to $D$. In the end we will
return $D_1\cup D_2\cup D_3$. We will call vertices have been selected
into a set $D_i$ \emph{green}, vertices that are dominated by a green
vertex but are not green themselves are called \emph{yellow} and all
vertices that still need to be dominated are called \emph{red}. In the
beginning, all vertices are marked red.

The first phase of our algorithm is similar to the first phase of the
algorithm of Lenzen et al.~\cite{lenzen2013distributed} for planar graphs.
It is a preprocessing step that leaves us with only vertices whose
neighborhoods can be dominated by a few other vertices. Lenzen et al.\
proved that if $G$ is planar, then there exist less than $3\gamma$
many vertices~$v$ such
that the open neighborhood $N(v)$ of $v$ cannot be dominated by $6$
vertices of $V(G)\setminus \{v\}$~\mbox{\cite[Lemma
6.3]{lenzen2013distributed}}. The lemma can be generalized to more
general graphs, see~\cite{amiri2019distributed}. We prove the
following lemma, which is stronger in the sense that the number of
vertices required to dominate the open neighborhoods is smaller than
in~\cite{lenzen2013distributed} and~\cite{amiri2019distributed}, at the cost of having slightly more vertices with that property.

\begin{lemma}\label{lenzen-improved}
  % Let $\nn$ be an integer strictly larger than $\nabla_1^B(G)$, the
  % edge density of a densest bipartite $1$-shallow minor of $G$.
  Let $\hat{D}$ be the set of vertices $v\in V(G)$ whose neighborhood
  cannot be dominated by $(2\nn-1)$ vertices of $D$ other than $v$,
  that is,

  \vspace{-4mm}
  \[
    \text{ $\hat{D}\coloneqq \{v\in V(G) :$ for all sets
      $A\subseteq D\setminus \{v\}$ with $N(v)\subseteq N[A]$ we have
      $|A|> (2\nn-1)\}$.}
  \]

  Then $|\hat{D}\setminus D| < \rho(G)\cdot\gamma$.
\end{lemma}

Remember that $\nn$ is an integer strictly larger than $\nabla_1^B(G)$, the
edge density of a densest bipartite $1$-shallow minor of $G$.
Additionally
$\rho(G)\le\chi(G)\leq 2\nabla_0(G)+1\leq 2\nabla_1+1$. The precise values
will be relevant for the planar case.

\begin{proof}
  Assume $D=\{b_1,\ldots,b_\gamma\}$.  Assume that there are
  $\rho(G)\cdot\gamma$ vertices
  $a_1,\ldots,a_{\rho(G)\cdot\gamma}\not\in D$ satisfying the above
  condition. Be definition of the Hall ratio we
  find an independent subset of the $a_is$ of size~$\gamma$. We can
  hence assume that $a_1,\ldots,a_{\gamma}$ are not connected by an
  edge. We proceed towards a contradiction.

  We construct a bipartite $1$-shallow minor $H$ of $G$ with the
  following \mbox{$2\gamma$} branch sets. For every
  \mbox{$i\le \gamma$} we have a branch set $A_i=\{a_i\}$ and a branch
  set
  $B_i=N[b_i]\setminus \big(\{a_1,\ldots, a_{\gamma}\}\cup
    \bigcup_{j<i}N[b_j]$ $\cup \{b_{i+1},\ldots, b_\gamma\}\big)$. Note that
  the $B_i$ are vertex disjoint and hence we define proper branch
  sets. Intuitively, for each vertex $v\in N(a_i)$ we mark the
  smallest $b_j$ that dominates $v$ as its dominator. We then contract
  the vertices that mark $b_j$ as a dominator together with $b_j$ into
  a single vertex. Note that because the $a_i$ are independent, the
  vertices $a_i$ themselves are not associated to a dominator as no
  $a_j$ lies in $N(a_i)$ for~$i\neq j$.  Denote by
  $a_1',\ldots, a_{\gamma}',b_1',\ldots, b_\gamma'$ the associated
  vertices of $H$. Denote by $A$ the set of the $a_i's$ and by $B$ the
  set of the~$b_j's$.  We delete all edges between vertices of
  $B$. The vertices of $A$ are independent by construction. Hence, $H$
  is a bipartite $1$-shallow minor of $G$.  By the assumption that
  $N(a_i)$ cannot be dominated by $2\nn-1$ elements of $D$, we
  associate at least~$2\nn$ different dominators with the vertices of
  $N(a_i)$. Note that this would not necessarily be true if~$A$ was
  not an independent set, as all $a_j\in N(a_i)$ would not be
  associated a dominator.

  Since $\{b_1,\ldots, b_\gamma\}$ is a dominating set of $G$ and by
  assumption on $N(a_i)$, we have that in $H$, every~$a'_i$ has at
  least $2\nn$ neighbors in $B$. Hence,
  $|E(H)| \ge 2\nn|V(A)| = 2\nn\gamma$. As $|V(H)|=2\gamma$ we
  conclude $|E(H)|\ge \nn|V(H)|$. This however is a contradiction,
  as by assumption~$\nn$ is strictly larger than $\nabla_1^B(G)$ the edge
  density of a densest
  bipartite $1$-shallow minor of $G$.
\end{proof}

% \pagebreak
Let us fix the set $\hat{D}$ for our graph $G$.
\begin{tcolorbox}
\vspace{-4mm}
    \begin{align*}
      \hat D\coloneqq \{v\in V(G) : \text{ for all } A\subseteq D\setminus \{v\}
      & \text{ with $N(v)\subseteq N[A]$} \text{ we have $|A|>(2\nabla-1)\}$.}
  \end{align*}
\end{tcolorbox}
\smallskip
Note that $\hat{D}$ cannot be computed by a local algorithm as we do
not know the set $D$. It will only serve as an auxiliary set in our
analysis.

\smallskip We define $D_1$ as the set of all vertices whose
neighborhood cannot be dominated by $2\nabla-1$ other vertices.
The first phase of the algorithm is to compute the set
$D_1$, which can be done in 2 rounds of communication.

\begin{tcolorbox}[colback=red!5!white,colframe=red!50!black]
  \vspace{-4mm}
  \begin{align*}
  D_1\coloneqq \{v\in V(G) : \text{ for all } A\subseteq V(G)\setminus \{v\}
  & \text{ with $N(v)\subseteq N[A]$ we have $|A|> (2\nabla-1)\}$.}
  \end{align*}
\end{tcolorbox}

\begin{lemma}\label{lem:size-D1}
  $D_1\subseteq \hat{D}$ and hence $|D_1\setminus D|\leq \rho(G)\cdot \gamma$.
\end{lemma}

\begin{proof}
  If the open neighborhood of a vertex $v$ cannot be dominated by $2\nabla-1$
  vertices from $V(G)\setminus\{v\}$, then in particular it cannot be
  dominated by $2\nabla-1$ vertices from $D\setminus\{v\}$.  Hence
  $D_1\subseteq \hat{D}$ and we can bound the size of $D_1$ by that of
  $\hat{D}$.
\end{proof}

\pagebreak
We mark the vertices of $D_1$ that we add to the dominating set in the
first phase of the algorithm as green, the neighbors of $D_1$ as
yellow and leave all other vertices red. Denote the set of red
vertices by~$R$, that is, $R=V(G)\setminus N[D_1]$.  For $v\in V(G)$
let $\Nr(v)\coloneqq N(v)\cap R$ and $\dr(v)\coloneqq |\Nr(v)|$
be the \emph{residual degree} of~$v$, that is, the number of neighbors
of $v$ that still need to be dominated.

\smallskip By definition of $D_1$, the neighborhood of every non-green
vertex can be dominated by at most~$2\nabla$ other vertices. This holds true
in particular for the subset $\Nr(v)$ of neighbors that still need to be
dominated.  Let us fix such a small dominating set for the red
neighborhood of every non-green vertex.

\begin{tcolorbox}
  For every $v\in V(G)\setminus D_1$, we fix
  $A_v\subseteq V(G)\setminus \{v\}$ such that:
    $$N_R(v)\subseteq N[A_v] ~~\text{and}~~ |A_v|\leq 2\nabla.$$

  Additionally, for vertices $v\in V(G)\setminus \hat{D}$, we enforce that
  $A_v\subseteq D\setminus \{v\}$.
\end{tcolorbox}

There are potentially many such sets $A_v$ -- we fix one such set
arbitrarily.
Let us stress that we cannot compute these sets
in a local algorithm as the sets $D$ and $\hat{D}$ are not known
to the algorithm. We only use these sets for our further argumentation.

\subsection{Limitations of the method}

\smallskip
We can apply the above approach to obtain a small set $D_1$ only if $\nabla_1$ is bounded by a constant. For example in graphs of bounded degeneracy in general the number of vertices that dominate the
neighborhood of a vertex can only be bounded by $\gamma(G)$.
Hence, the approach based on covers and pseudo-covers that is employed
in the following cannot be extended to degenerate  graph classes. Below
we show an example where this is the case.

\begin{example}
Let $G(\gamma,m)$ be the graph with vertices $v_i$ for $1\leq i\leq \gamma$,
$w^j$ for $1\leq j\leq m$ and $s_i^j$ for $1\leq i\leq \gamma, 1\leq j\leq m$.
We have the edges $\{v_1, w^j\}$ for $1\leq j\leq m$, hence $v_1$
dominates all~$w^j$. We have the edges $\{w^j, s_i^j\}$ for all $1\leq i\leq \gamma,
1\leq j\leq m$, hence, the $s_i^j$ are neighbors of $w^j$. Finally,
we have the edges $\{v_i, s_i^j\}$, that is, $v_i$ dominates the $i$th
neighbor of $w_j$ (see \cref{fig:example}). Hence, for $m>\gamma$,
$G(\gamma, m)$ has a dominating set of size
$\gamma$ and $m$ vertices whose neighborhood can be dominated
only by~$\gamma(G)$ vertices.
% \cref{lem:neighborhood-dom1} implies
% that $\gamma < 2\nabla_1$, and as we can choose~$m$ arbitrary
% large, we cannot usefully apply \cref{lem:neighborhood-dom1}.
Note that $G(\gamma,m)$ is
\mbox{$2$-degenerate}. As we can choose~$m$ arbitrary
large, we cannot usefully apply the method based on
\cref{lem:size-D1} for degenerate classes in general.

\vspace{-3mm}
\begin{center}
  \begin{figure}[h]
    \center
    \includegraphics[scale=0.3]{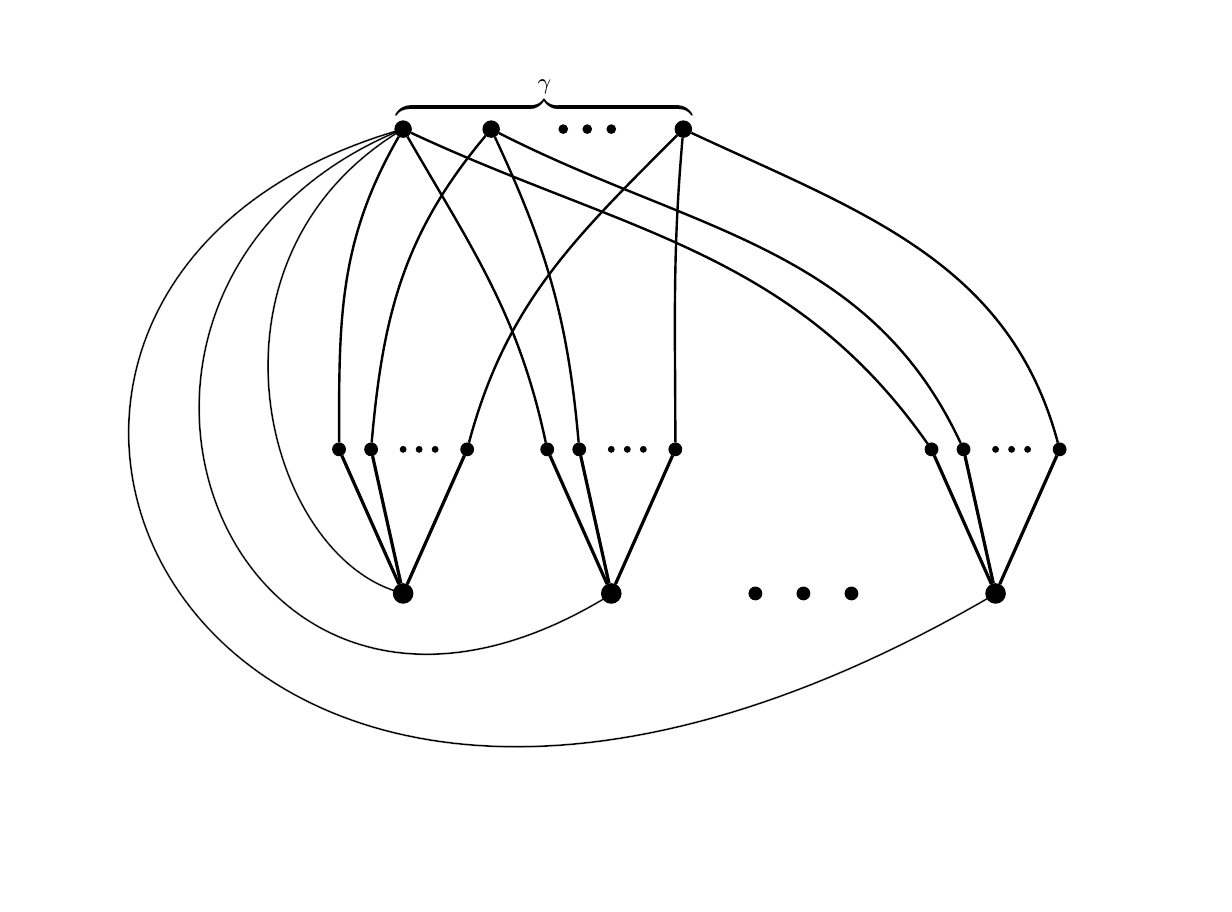}

    \vspace{-3mm}
    \caption{ A $2$-degenerate graph, where for many $v\in V(G)$ the set $N(v)$ can only be dominated by at least $\gamma$ vertices different from $v$. }
  \end{figure}\label{fig:example}
\end{center}
\end{example}

% !TEX root = main.tex

\section{Phase 2: Reducing residual degrees -- pseudo-covers and domination sequences}\label{sec:phase2}

After the first phase of the algorithm we have established the situation
that for every vertex \mbox{$v\in V(G)\setminus D_1$} the residual neighborhood
$\Nr(v)$ is dominated by set $A_v\subseteq V(G)\setminus\{v\}$
of size at most $2\nabla$.  For all vertices of
$V(G)\setminus \hat{D}$ we have chosen $A_v\subseteq D\setminus\{v\}$.
Observe that the set $\bigcup_{v\in V(G)}A_v$ has very good
domination properties. First, already the sets $A_v$ for $v\in D$
dominate almost all vertices that remain to be dominated, except possibly
the vertices of $D$ itself: We have $R\subseteq \bigcup_{v\in D} N_R[v]
= D\cup \bigcup_{v\in D} N_R(v)=D\cup \bigcup_{v\in D} A_v$, hence
$R\setminus D\subseteq \bigcup_{v\in D}A_v$.
Second, $\bigcup_{v\in V(G)}A_v$ is small, as
$|\bigcup_{v\in V(G)}A_v|\leq
\sum_{v\in V(G)} |A_v|
=\sum_{v\in \hat{D}}|A_v|+ \sum_{v\in V(G)\setminus\hat{D}}|A_v|$.
So~$|\bigcup_{v\in V(G)}A_v|
\leq (\rho(G)+1)\gamma\cdot 2\nabla+\gamma
\in \Oof(\gamma)$.

\pagebreak
In the second phase of
the algorithm we aim to find a good approximation of the sets $A_v$.
We follow the approach of Czygrinow et al.~\cite{czygrinow2018distributed}
and define \emph{pseudo-covers}, which describe candidate
vertices for the sets $A_v$. We will then consider a selection process that
can be carried out in parallel for all vertices, which is based on the
definition of \emph{domination sequences}, and allows to select
a bounded number of candidate vertices. The domination properties
of the selected vertices are worse than that of the sets $A_v$, however,
at the end of the second phase we will be in the situation that the
residual degree of each vertex is bounded by an absolute constant
depending only on the graph class under consideration.\vspace{-2mm}

\subsection{Pseudo-covers}

Following the presentation of~\cite{czygrinow2018distributed}, we name and fix
the following constants for the rest of this article. The reason to choose
the constants as given will become clear in the course of the proof.

\begin{tcolorbox}
\hfill
\begin{tabular}{l l}
$\kappa$       & $\coloneqq~ \max\{2\nabla_0,2\nn\}$,\\
$\lambda$  & $\coloneqq~ 1/\kappa$,
\vspace{-1mm}
\end{tabular}
\hfill
\begin{tabular}{l l}
$\mu$    & $\coloneqq~ 2\kappa/\lambda=2\kappa^2$,\\
$\nu$       & $\coloneqq~ k\mu = 2\kappa^3$.
\end{tabular}
\hfill~
\end{tcolorbox}

\begin{definition}
A vertex $z\in V(G)$ is \emph{$\lambda$-strong} for a vertex set $W\subseteq V(G)$ if $|N[z]\cap W|\geq \lambda|W|$.
\end{definition}

The following is the key definition by Czygrinow et al.~\cite{czygrinow2018distributed}.

\begin{definition}
  A \emph{pseudo-cover} (with parameters $\kappa, \lambda, \mu, \nu$)
  of a set $W\subseteq V(G)$ is a
  sequence $(v_1,\ldots, v_m)$ of vertices
  such that for every $i\le m$ we have:\\[-6mm]
  \begin{itemize}
    \item $m\leq \kappa$,\\[-5mm]
    \item $v_i$ is $\lambda$-strong for $W\setminus\bigcup_{j<i}N[v_j]$,\\[-5mm]
    \item $|N[v_i]\cap (W\setminus\bigcup_{j<i} N[v_j])|\geq \mu$, and\\[-5mm]
    \item $|W\setminus \bigcup_{j\le m}N[v_j]|\leq \nu$.
  \end{itemize}
\end{definition}
Intuitively, all but at most $\nu$ elements of the set $W$ are covered by the $(v_i)_{i\le m}$.
Additionally, each element $v_i$ of the pseudo-cover dominates both an
$\lambda$-fraction of the part of $W$ that is not yet dominated by the $v_j$
for $j<i$, and at least $\mu$ elements.
Note that with our choice of constants, if there are more than $\nu$ vertices not
covered yet, any vertex that covers a $\lambda$-fraction of what remains also
covers at least $\mu$ elements.

The next lemma shows how to derive the existence of pseudo-covers from
the existence of small dominating sets.

\begin{lemma}\label{lem:cover-to-pseudo-cover}
  Let $W\subseteq V(G)$ be of size at least
  $\nu$ and let~$Z$ be a dominating set of $W$ with~$\kappa$ elements.
  There exists an ordering of the vertices of $Z$ as $z_1,\ldots, z_\kappa$
  and $m\leq \kappa$ such that $(z_1,\ldots, z_m)$ is a pseudo-cover of $W$.
\end{lemma}
\begin{proof}
  We build the order greedily by induction. We order the elements by neighborhood size, while removing the neighborhoods of the previously ordered vertices. More precisely, assume that $(z_1,\ldots,z_i)$ have been defined for \mbox{some~$i\ge 0$}. We then define $z_{i+1}$ as the element that maximizes $|N[z] \cap (W \setminus \bigcup_{j\le i}N[z_j])|$.

  Once we have ordered all vertices of $Z$, we define $m$ as the maximal integer not larger than $\kappa$ such that for every $i \le m$ we have:
  \begin{itemize}
    \item $z_i$ is $\lambda$-strong for $W\setminus\bigcup_{j<i}N[z_j]$, and
    \item $|N[z_i] \cap (W \setminus \bigcup_{j\le i}N[z_j])| \ge \mu$.
  \end{itemize}

  This ensures that $(z_1,\ldots, z_m)$ satisfies the first 3 properties of a
  pseudo-cover of $W$. It only remains to check the last one.
  To do so, we define $W' \coloneq W \setminus\bigcup_{i\le m}N[z_i]$. We want to prove
  that $|W'| \le \nu$. Note that because $Z$ covers~$W$, if $m=\kappa$ we
  have $W'=\emptyset$ and we are done. We can therefore assume
  that $m<\kappa$ and $W'\neq \emptyset$. Since $Z$ is a dominating set of $W$,
  we also know that $(z_{m+1},\ldots z_\kappa)$ is a dominating set of $W'$,
  therefore there is an element in $(z_{m+1},\ldots z_\kappa)$ that
  dominates at least a $1 / \kappa$ fraction of $W'$. Thanks to the
  previously defined order, we know that~$z_{m+1}$ is such an element.
  Since $\lambda = 1/\kappa$, it follows that~$z_{m+1}$ is $\lambda$-strong
  for $W'$.
  This, together with the definition of $m$, we have that $|N[z_i] \cap (W \setminus \bigcup_{j\le i}N[z_j])| < \mu$ meaning that $|N[z_{m+1}] \cap W'| < \mu$. This implies that $|W'|/\kappa < \mu$. And since $\mu = \nu/\kappa$, we have $|W'|<\nu$.
  Hence, $(z_1,\ldots, z_m)$ is a pseudo-cover of $W$.
\end{proof}

%We will apply the pseudo-covers to neighborhoods
%of vertices.

While there can exist unboundedly many dominating sets for a set $W\subseteq V(G)$,
a nice observation of Czygrinow et al.\ was that the number of
pseudo-covers is bounded whenever the input graph excludes some
biclique $K_{s,t}$ as a subgraph. We do not state the result in this
generality, as it leads to enormous constants. Instead, we focus on the
case where small dominating sets $A_v$ exist, implying that $\nabla_0(G)$,
and therefore $\kappa$, are bounded.

\begin{lemma}\label{lem:num-high-degree}
  Let $W\subseteq V(G)$ of size at least $\mu$.
  Then there are less than $\kappa^2$ vertices that are
  $\lambda$-strong for $W$.
\end{lemma}
\begin{proof}
  Assume that there is such a set $W$ with $|W|\ge \mu$ and
  $c$ many vertices that are $\lambda$-strong for~$W$.
  Let $H$ be the subgraph of $G$ induced by $W$ and the $\lambda$-strong vertices.
  We first have that $|V_H| \le |W|+c$. Second we have that
  $|E_H|\ge \lambda|W|c-\nabla_0c$, because there are $c$ vertices that have
  degree at least $\lambda|W|$ and there are at most $\nabla_0c$ many vertices
  between them.

\pagebreak
  We then have that $|E_H|\le \nabla_0|V_H|$ hence
  $\lambda|W|c-\nabla_0c\le \nabla_0(|W|+c)$ from which we
  derive~$c(\lambda |W|-2\nabla_0)\le \nabla_0|W|$.
  Now, using that $|W|>\mu\ge 2k^2$ we have $\lambda|W|>2k>4\nabla_0$.
  Hence $\lambda|W|-2\nabla_0 \ge \lambda|W|/2$.
  We can finally deduce that $c(\lambda|W|/2)\le \nabla_0|W|$ and therefore we
  have that $c\le 2\nabla_0/\lambda=\kappa^2$.
\end{proof}

This leads quickly to a bound on the number of pseudo-covers.

% \alex{old version:}
%   \begin{lemma}\label{lem:num-pseudo-covers-old}
%     For every $W\subseteq V(G)$ of size at least $\ell$, the number of
%     pseudo-covers is less than~$(4\nabla_1/\alpha)^k$.
%   \end{lemma}
%
%   The proof of the lemma is exactly as the proof of Lemma~7 in the
%   presentation of Czygrinow et al.~\cite{czygrinow2018distributed},
%   we reprove it for the sake of completeness.
%
%
%   \begin{proof}
%     Let $W$ a set of size at least $\ell$. For every $i\le k$, we define $C_i$ as the
%     set of partial pseudo-covers of~$W$ of size at most $i$, that is,
%     all sets of at most $i$ vertices that can be extended to a pseudo-cover
%     of~$W$.
%     So $C_k$ is the set of pseudo-covers of $W$ while $C_1$ only contains
%     $\alpha$-strong vertices for $W$.
%
%     \cref{lem:num-high-degree} implies that $|C_1|< 4\nabla_1/\alpha$.
%     \cref{lem:num-high-degree} also implies that for every $i<k$, we have
%     $|C_{i+1}| < |C_{i}|\cdot (4\nabla_1/\alpha)$. We therefore conclude that
%     $|C_k| < (4\nabla_1/\alpha)^k$.
%   \end{proof}

\begin{lemma}\label{lem:num-pseudo-covers}
  For every $W\subseteq V(G)$ of size at least $\mu$, the number of
  pseudo-covers is less than~$\kappa^{2\kappa}$.
\end{lemma}

The proof of the lemma is exactly as the proof of Lemma~7 in the
presentation of Czygrinow et al.~\cite{czygrinow2018distributed},
we reprove it for the sake of completeness.

\begin{proof}
  Let $W$ a set of size at least $\mu$. For every $i\le \kappa$, we define $C_i$ as the
  set of partial pseudo-covers of~$W$ of size at most $i$, that is,
  all sets of at most $i$ vertices that can be extended to a pseudo-cover
  of~$W$.
  So $C_\kappa$ is the set of pseudo-covers of $W$ while $C_1$ only contains
  $\lambda$-strong vertices for $W$.

  \cref{lem:num-high-degree} implies that $|C_1|< \kappa^2$.
  \cref{lem:num-high-degree} also implies that for every $i<\kappa$, we have
  $|C_{i+1}| < |C_{i}|\cdot \kappa^2$. We therefore conclude that
  $|C_k| < (\kappa^2)^\kappa$.
\end{proof}

\begin{tcolorbox}
  We write $\Tt(v)$ for the set of all pseudo-covers
  of $N(v)$ and $\Pp(v)$ for the set of all vertices that appear in a
  pseudo-cover of $N(v)$.
\end{tcolorbox}

The proof of \cref{lem:num-pseudo-covers} also bounds the number $\Pp(v)$.
\begin{corollary}\label{cor:nb-dominators}
For every $v\in V(G)$ with $|N_R(v)|> \mu$, we have
  % $|\Tt(v)|< k^{2k}$ and
  $|\Pp(v)|\le \kappa^{2\kappa}$.
\end{corollary}

%We recall all fixed parameters for easy to find reference.
%\textcolor{red}{Motivate the definition of these parameters in
%the construction of the pseudo covers.}
%
%\bigskip
%\begin{tcolorbox}
%\begin{tabular}{l l}
%- $G$ & :~ fixed graph. \\
%- $\gamma$ & :~ $\gamma(G)$.\\
%- $\nabla_1$ & :~ $\nabla_1(G)$. \\
%- $t$ & :~ smallest integer such that $G$
%excludes $K_{t,t}$ as a subgraph\\
%- $D_1$ & :~ defined and computed in \cref{lem:neighborhood-dom1}.\\
%- $D$ & :~ fixed dominating set of $G$ of size $\gamma$ (not computed).\\
%- $\hat{D}$ & :~ defined in \cref{lem:neighborhood-dom2} (not computed).\\
%- $D'$ & :~ $D \cup \hat{D}$ (not computed).
%\end{tabular}
%\end{tcolorbox}

% For the rest of this article we fix
% $\alpha\coloneqq 1/k$, $\ell=8\nabla_1/\alpha^2+1=4k^3+1$ and $q=4k^4$.
% These constants will be needed in the following definitions of pseudo-covers
% and domination sequences.
% We also fix the following constants to follow the presentation
% of~\cite{czygrinow2018distributed}.

%\input{5-covers}

% !TEX root = main.tex

\subsection{Domination sequences}

We now turn to the use of pseudo-covers.
%Since the biclique $K_{s,t}$ has
%$s\cdot t$ edges and $s+t$
%vertices we conclude that there are constants
%$s,t\leq 2\nabla_0+1\leq 2\nabla_1+1$ such that $K_{s,t}\not\subseteq G$.
%\begin{tcolorbox}
%Let $s\le t$ be such that $K_{s,t}\not\subseteq G$, which implies that $s\leq t\leq 2\nabla_0+1$.
%\pom{No, it does not imply!
%	We should say: for example, we can let $s$ and $t$  be any positive integers with $\frac{1}{\nabla_0(G)}>\frac1s+\frac1t$ (such as $s=t=2\nabla_0+1$), or $s=t=2$ if  $G$ has girth at least $5$, or $s=t=3$ if $G$ is planar, etc.\\
%Note that for general graphs, the best is probably to let $s=\lfloor\nabla_0(G)\rfloor+1$.
%}
%\end{tcolorbox}
%
%\smallskip
We aim to carry out an iterative process in parallel
for all vertices \mbox{$v\in V(G)$} with a sufficiently large
residual neighborhood $\Nr(v)$.
%Remember that we defined $R := V(G) \setminus N[D_1]$.

\begin{definition}\label{def:dom-sequence}
  For any vertex $v\in V(G)$, a {\em $\kappa$-dominating-sequence} of $v$ is a
  sequence $(v_1,\ldots,v_m)$ (without repetition) for which we can define
  sets $B_1,\ldots,B_m$ such that:
  \begin{itemize}
    \item $v_1=v$, $B_1 \subseteq \Nr(v_1)$,
    \item for every $i\le m$ we have $B_{i} \subseteq (\Nr(v_{i})\cap B_{i-1})$,
    \item $|B_{i}|\geq \kappa^{s -i}(t+s -i+(s -i)\nu)$
    \item and for every $i\le m$ we have $v_i\in \Pp(v_{i-1})$.
  \end{itemize}
  A $\kappa$-dominating-sequence $(v_1,\ldots,v_m)$ is {\em maximal} if there is no
  vertex $u$ such that $(v_1,\ldots,v_m,u)$ is a $\kappa$-dominating-sequence.
\end{definition}

% Remember $s,t$ from the excluded biclique $K_{s,t}$.
Note that this definition requires $|\Nr(v)|\ge \kappa^{s -1}(t+s -1+(s -1)\nu)$.
For a vertex~$v$ with a too small residual neighborhood, there are no
$\kappa$-dominating-sequences.
We show two main properties of these dominating-sequences.
First, \cref{lem:max-dom-sequence} shows that a maximal dominating sequence must
encounter $D\cup\hat{D}$ at some point. Second, with \cref{lem:shape-sequences,lem:small-D-hat,lem:inclusion-D-hat},
we show that collecting all ``end points'' $v_m$ of all maximal dominating
sequences results in a set $D_2$ of size linear in the size of $D$. While $D$
cannot be computed, we can compute $D_2$.

\begin{lemma}\label{lem:max-dom-sequence}
  Let $v$ be a vertex and let
  $(v_1,\ldots, v_m)$ be a maximal $\kappa$-dominating-sequence of $v$. Then $m<s$ and
  $(D\cup\hat{D})\cap \{v_1,\ldots, v_m\}\neq \emptyset$.
\end{lemma}
\begin{proof}
  First, assume that $v_1,v_2,\ldots,v_m$ is a maximal
  $\kappa$-dominating-sequence with $m\ge s$.
  By definition, every~$v_i$ with $i\le s$ is connected to every vertex of $B_s$.
  For every $1\leq i\leq s$ we have $|B_i|\geq t$
  and in particular
  $|B_s|\ge t$. This shows that the two sets
  $\{v_1,\ldots,v_s\}$ and $B_s$ form a $K_{s,t}$ as a subgraph in $N^2[v]$.
  Since $K_{s,t}$ is excluded as a subgraph in~$G$, the process must stop
  having performed at most $s-1$ rounds.

  We now have $m<s$ and to prove the second statement we assume, in order to reach
  a contradiction, that
  $(D\cup\hat{D})\cap\{ v_1,\ldots, v_m\}=\emptyset$.
  We have that $B_m \subseteq N(v_m)$, and remember that as~$v_m$ is not
  in~$\hat{D}$, we have that $\Nr(v_m)$ can be
  dominated by at most $\kappa$ elements of $D$.

  By \cref{lem:cover-to-pseudo-cover},
  we can derive a pseudo-cover $S=(u_1,\ldots,u_j)$ of
  $\Nr(v_m)$, where $j\le \kappa$ and every $u_i$ is an element of~$D$.
  Let $X$ denote the set (of size at most $\nu$) of vertices of $B_m$ not covered by $S$.
  As $S$ contains at most $\kappa$ vertices there must exist a
  vertex $u$ in~$S$ that covers at least a $1/\kappa$ fraction of
  $B_m\setminus X$.
  By construction, we have that $|B_m| \ge \kappa^{s-m}\cdot(t+s-m+(s-m)\nu)\ge \kappa(t+\nu)$
  because $m<s$. Therefore $|B_m\setminus X| \ge \kappa$ and we have
  $$|\Nr[u]\cap B_{m}|\geq \frac{|B_m|-\nu}{\kappa}
  \geq\frac{\kappa^{s-m}(t+s-m+(s-m)\nu) -\nu}{\kappa},$$
  hence
  $$|\Nr[u]\cap B_{m}| \geq\frac{\kappa^{s-m}(t+s-m+(t-m-1)\nu)}{\kappa} \geq
  \kappa^{s-m-1}(t+s-m+(t-m-1)\nu),$$
  and therefore
  $$ |\Nr(u)\cap B_{m}| \geq |\Nr[u]\cap B_{m}|-1 \geq \kappa^{s-m-1}(t+s-m-1+(t-m-1)\nu).$$
  So we can continue the sequence $(v_1,\ldots,v_m)$ by defining
  $v_{m+1}\coloneqq u$; there is no repetition since by hypothesis
  $D\cap\{ v_1,\ldots, v_m\}=\emptyset$, and by construction
  $u\in D$.

  In conclusion if $(v_1,\ldots,v_m)$ is a maximal
  sequence, it contains an element of $D$ or $\hat{D}$.
\end{proof}

Our next goal is to show that there are not many elements $v_m$
(which are the elements that we pick into the set $D_2$).

%\alex{Should we keep the example or refer to Section 6.2?}
%This is illustrated in the following example and formalized right after that.
%
%\begin{example}\label{ex:sequence}
%  Consider the case of planar graphs. Since these graph exclude $K_{3,3}$,
%  i.e.~$t=3$, we have that every maximal sequence is of length 1 or 2.
%  For every $v$ of sufficiently large neighborhood we consider every
%  maximal $k$-dominating-sequence $(v_1,v_s)$ of $v$.
%  We then add $v_s$ to the set $D_2$. We want to show that $|D_2|$ is
%  linearly bounded by $|D'|$ and hence by $\gamma(G)$.
%
%  If $s=1$, then we have $v_s\in D'$ and we are good.
%
%  If $s=2$, we have two possibilities. If $v_2$ is in $D'$, we are good.
%  If however, $v_2$ is not in $D'$, then $v_1$ is in
%  $D'$. Additionally, $v_2$ is in some pseudo-cover $S$ of $v_1$,
%  i.e.~$v_2\in \Pp(v_1)$.
%
%  By \cref{cor:nb-dominators}, we have $|\Pp(v_1)|\le (2k)^{2k+1}$ (and
%  in fact this number is much smaller in the case of planar graphs).
%  Therefore we have  $|D_2| \le ((2k)^{2k+1}+1)|D'|$.
%\end{example}
%
%We generalize the ideas of \cref{ex:sequence}, by explaining what
%a ``few possible choices''  in the discussion before \cref{def:dom-sequence}
%means.

\begin{lemma}\label{lem:shape-sequences}
  For any maximal $\kappa$-dominating-sequence $(v_1,\ldots,v_m)$,
  and for any $i\le m-1$, we have that
  \begin{itemize}
    \item $v_{i+1}\in \Pp(v_i)$, and
    \item $|\Nr(v_i)|\ge \mu$.
  \end{itemize}
\end{lemma}
\begin{proof}
  By construction we have $v_{i+1}\in \Pp(v_i)$, furthermore
  $|B_{i}|\geq \kappa^{t -i}(2t -i+(t -i)\nu) \ge \nu >\mu$,
  and~$B_i\subseteq \Nr(v_i)$.
  %We conclude with~\cref{cor:nb-dominators}.
\end{proof}

Now, for every $v\in V(G)$ we compute all maximal $\kappa$-dominating-sequences
starting with $v$.
Obviously, as every $v_i$ in any $\kappa$-dominating-sequences of $v$ dominates some
neighbors of $G$, we can locally compute these steps after having
learned the $2$-neighborhood $N^2[v]$ of every vertex in two rounds
in the LOCAL model of computation.

\begin{tcolorbox}[colback=red!5!white,colframe=red!50!black]
  We define $D_2$ as the set of all $u\in V(G)$ such that there is some vertex
  $v\in V(G)$, and some maximal $\kappa$-dominating-sequence $(v_1,\ldots,v_m)$ of
  $v$ with $u=v_m$.
\end{tcolorbox}

We now take a look at the size of $D_2$.
For a set $W\subseteq V(G)$ we write $\Pp(W) = \bigcup_{v\in W}\Pp(v)$.
Remember that the definition of $\Pp(v)$ requires that $|\Nr(v)|>\mu$. We simply
extend the notation with $\Pp(v)=\emptyset$ if $|\Nr(v)|\le \mu$.
We then define:
\[\Pp^{(1)}(W)\coloneqq \Pp(W)\]
for $1<i <s$
\[\Pp^{(i)}(W)\coloneqq \Pp(\Pp^{(i-1)}(W))\]
and, for $1\le i \le s$
\[\Pp^{(\leq i)}(W)\coloneqq \bigcup_{1\leq j\leq i}\Pp^{(j)}(W).\]

  % for the set of all vertices that appear in some $S\in \Tt $.
  % For a set $U\subseteq V(G)$ let \[\Tt(U)\coloneqq \bigcup_{v\in U}\Tt .\]
  % For a set $\Ss$ of pseudo-covers, again with a slight
  % abuse of notation, we define \[\Tt(\Ss)\coloneqq \Tt(\bigcup \Ss).\]
  % We now define
  % \[\Tt^{(1)}(U)\coloneqq \Tt(U)\]
  % and  for $i>1$
  % \[\Tt^{(i)}(U)\coloneqq \Tt(\Tt^{(i-1)}(U)).\]
  % Finally, $\Tt^{(\leq k)}(U)\coloneqq \bigcup_{1\leq i\leq k}\Tt^{(i)}(U)$.

We are now ready to prove that $D_2$ is small.

\begin{lemma}\label{lem:inclusion-D-hat}
  $D_2 \subseteq \Pp^{(\le s)}(D\cup\hat{D})$.
\end{lemma}
\begin{proof}
  Using \cref{lem:shape-sequences} repetitively, for every
  $\kappa$-dominating-sequence $(v_1,\ldots,v_m)$ we have that\linebreak
  $v_m \in \Pp^{(\le s)}(v_1)$, and, more generally, for every $i\le m$, we have that
  $v_m\in \Pp^{(\le s)}(v_i)$. Now the statement follows from \cref{lem:max-dom-sequence}.
\end{proof}

\begin{lemma}\label{lem:small-D-hat}
  $|D_2| \le (\kappa^{2s\kappa}(\rho(G)+1)\gamma$.
\end{lemma}
\begin{proof}
  \cref{cor:nb-dominators} gives us that $|\Pp(v)|\le \kappa^{2\kappa}$ for every
  $v\in V(G)$ with $|N(v)|> \mu$.
  %This with \cref{lem:shape-sequences} implies that
  As $\Pp(W) \le \sum\limits_{v\in W} |\Pp(v)|$,
  we have $P(W)\le |W|\cdot \kappa^{2\kappa}$.
  A simple induction yields that for $i\le t$,
  \[ |\Pp^{(\le i)}(W)|\leq c^i|W|, \] where $c=\kappa^{2\kappa}$.
  With \cref{lem:inclusion-D-hat} we conclude
  \[|D_2| \le \kappa^{2s\kappa} \cdot|D\cup\hat{D}|.\]
  We conclude with \cref{lenzen-improved}, stating that
  $|\hat{D}\setminus D|\leq \rho(G)\cdot \gamma$.
\end{proof}

We update the set $R$ of vertices that still need to be dominated
as $V(G)\setminus N[D_1\cup D_2]$
and the residual neighborhoods $\Nr(v)= N(v)\cap R$
and residual degrees $\dr(v) =|\Nr(v)|$. We
prove next that $\dr(v)$ is bounded by a constant.

\begin{lemma}\label{lem:smalldegree}
  For every vertex $v\in V(G)$ we have $\dr(v)< \kappa^{s-1}(t+s-1+(s-1)\nu)$.
\end{lemma}
\begin{proof}
  Assume, for the sake of reaching a contradiction, that there is a vertex $v$
  satisfying \linebreak$\dr(v)   \ge  \kappa^{s -1}(t+s -1+(s -1)\nu)$ and let $B_1\coloneqq \Nr(v)$. Note that $v\not\in D_1\cup D_2$, as the residual degree of vertices from this
  set is $0$.
  Exactly as in the proof of \cref{lem:max-dom-sequence}, since $v\not\in D_1$,
  we have that $B_1$ can be dominated by at most $\kappa$ elements. Hence by
  \cref{lem:cover-to-pseudo-cover}, we can derive a
  pseudo-cover $S=(u_1,\ldots,u_j)$ of
  $B_1$, where $j\le \kappa$. This leads to the existence of some vertex~$u$ in~$S$
  that covers at least a $1/\kappa$ fraction of $B_1\setminus X$ for some $X$ of size
  at most $\nu$. This yields a
  vertex~$v_2$, and a set~$B_2$.

  We can then continue and build a maximal $k$-dominating-sequence
  $(v_1,\ldots v_m)$ of $v$. By construction, this sequence has the property
  that every $v_i$ dominates some elements of $B_1$. This is true in particular
  for $v_m$, but also we have that $v_m\in D_2$, hence a contradiction.
\end{proof}

\begin{tcolorbox}
  Let $\Dr\coloneqq~\kappa^{s-1}(t+s-1+(s-1)\kappa^3)$.
\end{tcolorbox}

As it remains to dominate the set $R$, let us fix a minimum dominating
set $D_R$ of size $\gamma_R$ for $R$.

\smallskip
\begin{tcolorbox}
\begin{itemize}
\item Let $D_R\subseteq V(G)$ be a minimum dominating
  set of $R$ and let $\gamma_R\coloneqq |D_R|$.
\item Let $\eta\in [0,1]$ be such that $|(D_1 \cup D_2)\cap D| =\eta\gamma$.
\end{itemize}
\end{tcolorbox}

\begin{lemma}\label{lem:size-DR}
$\gamma_R\leq (1-\eta)\gamma$.
\end{lemma}
\begin{proof}
$D\setminus(D_1\cup D_2)$ is a dominating set for $R$, hence
$|D_R|\leq |D\setminus(D_1\cup D_2)|$.
\end{proof}

As every vertex of $D_R$ can dominate at most $\Dr+1$ vertices (its $\Dr$ residual neighbors and itself), we have the following corollary.

\begin{corollary}\label{cor:size-R}
  $|R|\leq (\Dr+1)\gamma_R$.
\end{corollary}

% !TEX root = main.tex

\section{Phase 3: LP-based approximation in graphs of bounded maximum degree}\label{sec:LP}

\subsection{LP-based approximation}
In the light of \cref{cor:size-R}, we could now simply choose
$R$ as the set $D_3$ to get a
constant factor approximation. We can improve the bounds however, by
proceeding with an LP-based approximation. The dominating
set problem can be formulated as an integer linear program~(ILP).
Note that it remains to dominate the set $R$, which leads to the
following ILP.

\[
  \begin{array}{l l l}
    \text{Minimize }    & \sum_{v\in V}x_v &\\
    \text{Subject to }\quad & \sum_{u\in N[v]}x_u \ge 1 \quad &\forall v\in R \\
                            & x_v\in \{0,1\}                   &\forall v \in V\\
  \end{array}
\]

By relaxing the condition that $x_v\in \{0,1\}$ to $x_v\in [0,1]\subseteq \mathbb{R}$,
we obtain the corresponding linear program (LP). By a result of
Bansal and Umboh~\cite{bansal2017tight} one can obtain a constant
factor approximation of a dominating set from a solution to the LP.
The proof can easily be adapted to the problem of approximating a
dominating set of the set $R$.

\begin{lemma}\label{lem:ds-factor}
  Assume $G$ has an orientation with maximum out-degree $d$.
  Let $\big(x_v\big)_{v\in V(G)}$ be a solution to the
  $R$-dominating set~LP. Let $H:=\{v\in V(G) : x_v\ge 1/(2d+1)\}$ and
  let $U:=\{v\in R : v\not\in N[H]\}$. Then $H\cup U$ dominates~$R$
  and has size at most $(2d+1)\cdot \gamma_R$.
\end{lemma}

Observe that when given the solution $\big(x_v\big)_{v\in V(G)}$ to the
$R$-dominating set~LP the lemma gives rise to a simple LOCAL algorithm.
First select all vertices $v$ with $x_v\geq 1/(2d+1)$ into a dominating set
and mark all their neighbors as dominated. Then select all non-dominated
vertices of $R$ into the dominating set. Clearly, $H\cup U$ is a dominating
set of $R$. The rest of this section is devoted to the proof of the claimed
approximation factor. The proof follows the presentation of Bansal and
Umboh~\cite{bansal2017tight} with the improved bounds of Dvo\v{r}\'ak~\cite{dvovrak2019distance} (presented in \cref{lem:ds-factor}).
As every solution to the ILP is also a solution
to the LP we have $\sum_{v\in V(G)}x_v\leq \gamma_R$.

Consider an orientation of $G$ such that the neighborhood of each vertex
$v$ is decomposed into $N^{in}(v)$ and $N^{out}(v)$, where $|N^{out}(v)|\le d$.

\begin{claim}
  For every vertex $v\in U$, we have
  $\big(\sum_{u\in N^{in}(v)}x_u \big)\ge d/(2d+1)$.
\end{claim}
\begin{proof}
  As $v$ is not in $H$, $x_v<1/(2d+1)$. As $v$ is not in $N(H)$, for every
  vertex $u\in N^{out}(v)$ we have $x_u<1/(2d+1)$. As $|N^{out}(v)|\le d$,
  and by the first LP condition
  $\big(\sum_{u\in N^{in}(v)}x_u \big)\ge 1- \frac{1}{2d+1} - \frac{d}{2d+1}
  \ge \frac{d}{2d+1}$.
\end{proof}

We can now bound the size of $U$ and $H$
\begin{claim}
  $|H\cup U| \le (2d+1)\sum_{v\in V}x_v$.
\end{claim}
\begin{proof}
  First, observe that
  $|H|\le (2d+1)\sum_{v\in H}\frac{1}{2d+1}
  \le (2d+1)\sum_{v\in H}(x_v)$.
  Then observe that\linebreak
  $|U|\le \frac{2d+1}{d}\cdot\sum_{v\in U} \frac{d}{2d+1}
  \le \frac{2d+1}{d}\sum_{v\in U}\sum_{u\in N^{in}(v)}x_u
  \le \frac{2d+1}{d}\sum_{u\in N^{in}(U)} (d\cdot x_u)$.
  Form which we conclude
  $|U|
  \le (2d+1) \sum_{u\in N^{in}(U)} x_u$.

  By definition of $U$, we have that $N(U)$ and $H$ are disjoint, this also
  holds for $H$ and $N^{in}(U)$, hence $|H\cup U| \le (2d+1)\sum_{v\in V}x_v\leq (2d+1)\gamma_R$.
\end{proof}

\subsection{Solving LPs locally}

As shown by Kuhn et al.~\cite{kuhn2006price} we can locally approximate
general covering LPs, in particular the above $R$-dominating set LP,
when the maximum
degree of the graph is bounded. More precisely, they show how to compute
a $\Delta^{1/r}$-approximation in $\Oof(r^2)$ rounds. Assuming for a
moment that $\Delta$ is bounded by an absolute constant we can
choose $r$ such that $\Delta^{1/r}=1+\e$,
hence $r=(\log \Delta)/(\log (1+\e))$, which is a constant depending
only on $\Delta$ and $\e$ in order to compute a
$(1+\e)$-approximation for the $R$-dominating set LP.

\begin{corollary}\label{cor:LP-approx-general}
  Assume $G$ has an orientation with maximum out-degree
  $d$. For every $\e>0$ we can
  compute a set $D'$ of size at most~$(2d+1)(1+\e)\gamma_R$ that dominates
  $R$ in $\Oof(\log \Delta/(\log (1+\e))$ rounds in the LOCAL
  model.
\end{corollary}

\subsection{From bounded residual degree to bounded degree}

It remains to establish the situation that the maximum degree $\Delta$
of our graph is bounded. As argued, we have $|R|\leq (\Dr+1)\gamma_R$.
As only
the vertices of $R$ need to be
dominated it suffices to keep only the vertices that have a
neighbor in $R$; other vertices are not useful as dominators.
Also, when two vertices $u,v\in V(G)\setminus R$ have exactly
the same neighbors in $R$, that is, $\Nr(u)=\Nr(v)$,
it suffices to keep one of $u$ and $v$.
Note that we can locally decide whether $\Nr(u)=\Nr(v)$.
For every set $N\subseteq R$ such that there is a vertex $v$
with $\Nr(v)=N$ we choose the one with the lowest identifier
as a representative. We construct the graph $G'$ consisting
of $R$ and all edges between vertices in $R$ as well as the
set of all representatives and a minimal set of edges such that
$\Nr(v)$ is equal in $G$ and $G'$ for all representatives $v$.
Hence in $G'$ we have $\Nr(u)\neq \Nr(v)$ for all $u\neq v\in V(G')\setminus R$.
As argued above, every $R$-dominating set in $G$ can
be transformed into an $R$-dominating set of the same size
in $G'$ (by choosing appropriate representatives) and
every $R$-dominating set in $G'$ is an $R$-dominating set in~$G$.
We can hence continue to work with the graph $G'$. In order to
avoid complicated notation we simply assume that $G=G'$.

Note that in general we could have $|V(G)|\in \Omega(2^{|R|})$.
When $\nabla_1(G)$ bounded, however,  it follows from Lemma 4.3 of \cite{gajarsky2017kernelization} that $|V(G)|\leq (4^{\nabla_1}+2\nabla_1)|R|$,
which is is linear
in~$|R|$. This is crucial for our further argumentation.

%\begin{lemma}\label{lem:smallX-general}
%  $|V(G)|\le (4^{\nabla_1}+2\nabla_1)|R|$.
%\end{lemma}
%
%We will provide an improved lemma with a proof for the planar case below.

\begin{corollary}\label{cor:size-g-general}
  $|V(G)| \le (4^{\nabla_1}+2\nabla_1)(\Dr+1)\gamma_R$.
\end{corollary}

\subsection{Conclusion of the algorithm}

Given any $\epsilon>0$ we now select all vertices with high degree
$\Gamma=\Gamma(\epsilon)$ into our
dominating set, where $\Gamma$ is chosen such that there exist at
most $\epsilon\gamma$ vertices of degree at least $\Gamma$.

\begin{tcolorbox}[colback=red!5!white,colframe=red!50!black]
  Let $\Gamma\coloneqq 4\nabla_1(4^{\nabla_1}+2\nabla_1)(\Dr+1)/\epsilon$ \quad and \quad
  $D_3^1\coloneqq \{v\in V(G) ~:~ d(v)>\Gamma\}$.
\end{tcolorbox}

\begin{lemma}\label{lem:size-D31}
  $|D_3^1|\le (\epsilon/2)\gamma_R$.
\end{lemma}
\begin{proof}
  We assume the opposite and count the number of edges of $G$.
  When we sum the degree of the vertices, we get twice the number of
  edges. Hence $2\cdot |E(G)| > 2\nabla_1(4^{\nabla_1}+2\nabla_1)(\Dr+1)\gamma_R$.
  Therefore, with \cref{cor:size-g-general},
  $|E(G)|> \nabla_1|V(G)|$, a contradiction.
\end{proof}

After picking $D_3^1$ into the dominating set, marking the neighbors of
$D_3^1$ as dominated and updating the set $R$, we can delete the
vertices of $D_3^1$. We are left with a graph of maximum degree~$\Gamma$.

\begin{tcolorbox}[colback=red!5!white,colframe=red!50!black]
  Given $\epsilon>0$, let $D_3^2$ be the set computed by the LOCAL algorithm of \cref{cor:LP-approx-general} with parameter $\epsilon/2$.
\end{tcolorbox}

Let $D_3=D_3^1 \cup D_3^2$. We already noted that the definition of $D_3$ implies that
$D_1\cup D_2\cup D_3$ is a dominating set of $G$. We now conclude the
analysis of the size of this computed set.

\begin{lemma}\label{lem:D3-LP}
  We have that $|D_3| \le (2\nabla_0+1)(1+\e)\gamma_R$.
\end{lemma}
\begin{proof}
By \cref{lem:orientations} $G$ has an orientation with
out-degree $d\leq \nabla_0$. By
\cref{cor:LP-approx-general} and \cref{lem:size-D31} we have
  $|D_3^2|\leq (2\nabla_0+1)(1+\e/2)\gamma_R$, and $|D_3^1|\le (\epsilon/2)\gamma_R$.
\end{proof}

Now our main theorem, \cref{thm:main-general}, follows by summing the sizes of
$D_1$, $D_2$ and $D_3$.

\begin{lemma}
$|D_1\cup D_2\cup D_3|\le 2(\nabla_0+1)(\kappa^{2s\kappa}+2)\gamma$.
Hence, putting  $c=2(\nabla_1+1)$, we have
\[
|D_1\cup D_2\cup D_3|\le \bigl(c^{2c^2}+c\bigr)\gamma.
\]
%$in \Oof\hspace{1pt}((4c)^{16c^2}\gamma)$
\end{lemma}
\begin{proof}
We have
\begin{align*}
	\rho(G)&\leq 2\nabla_0(G)+1&\text{(by \cref{lem:bounds})}\\
	|D_1|&\leq \rho(G)\gamma\leq (2\nabla_0+1)\gamma&\text{(by \cref{lem:size-D1})}\\
	|D_2| &\le \kappa^{2s\kappa}(\rho(G)+1)\gamma\leq \kappa^{2s\kappa}(2\nabla_0+2)\gamma&\text{(by \cref{lem:small-D-hat})}\\
	\intertext{Last, by setting $\epsilon=1$ in \cref{lem:D3-LP} we have}
	|D_3|&\leq (2\nabla_0+1)(1+\e)\gamma_R\leq 2(2\nabla_0+1)\gamma.
\end{align*}
We conclude, as  $\kappa\le 2\nabla_1+2$ and $s\leq 2\nabla_1+1$.
%By \cref{lem:bounds} we have $\rho(G)\leq 2\nabla_0(G)+1$.
%By \cref{lem:size-D1} we we have $|D_1|\leq \rho(G)\gamma\leq (2\nabla_0+1)\gamma$.
%According to \cref{lem:small-D-hat} we have $|D_2| \le \kappa^{2s\kappa}(\rho(G)+1)\gamma\leq \kappa^{2s\kappa}(2\nabla_0+2)\gamma$. Finally, by setting $\epsilon=1/2$ in
%\cref{lem:D3-LP} we have $|D_3|\leq (2\nabla_0+1)(1+2\e)\gamma_R\leq
%2(2\nabla_0+1)\gamma$.

%By  we conclude $|D|\leq (2\nabla_0+1)$
\end{proof}

% !TEX root = main.tex

\section{Alternative Phase 3: Greedy domination}\label{sec:greedy-be}

We now consider an alternative approach for the third phase, which does
not improve the approximation factor, however, is conceptually much
simpler and interesting in its own. Recall that we bounded
$\Dr$ as $\kappa^{s-1}(t+s-1+(s-1)\nu)$, which is a bound on the residual
degree $\dr(v)$ of all vertices.

\newcommand{\dR}{\Dr}
\newcommand{\ddR}{\dr}
\newcommand{\pick}{P}

We simulate the classical greedy algorithm, which in each
round selects a vertex of maximum residual degree. Here, we let all
non-dominated vertices that have a neighbor of maximum residual degree
choose such a neighbor as its dominator (or if they have maximum
residual degree themselves, they may choose themselves). In general
this is not possible for a LOCAL algorithm, however, as we established
a bound on the maximum degree we can proceed as follows. We
let~$i={\dR}$. Every red vertex that has at least one neighbor of residual
degree ${\dR}$ arbitrarily picks one of them and elects it to the
dominating set. Then every vertex recomputes its residual degree and
$i$ is set to ${\dR}-1$. We continue until $i$ reaches $0$ when all
vertices are dominated. More formally, we define several sets as
follows.

\smallskip
\begin{tcolorbox}[colback=red!5!white,colframe=red!50!black]
  For ${\dR}\geq i\geq 0$,  for every $v\in R$ in parallel:\\[2mm]
  If there is some $u$ with ${\ddR}(u)=i$ and ($\{u,v\}\in E(G)$ or $u=v$), then\\
  \mbox{ } $\dom_i(v)\leftarrow \{u\}$ (pick one such $u$ arbitrarily),\\
  \mbox{ } $\dom_i(v)\leftarrow \emptyset$ otherwise.
  \begin{itemize}
    \item $R_i \leftarrow R$ \hfill \textit{\small What currently remains to be dominated}
    \item $\pick_i \leftarrow \bigcup\limits_{v\in R} \dom_i(v)$ \hfill \textit{\small What we pick in this step}
    \item $R \leftarrow R \setminus N[\pick_{i}]$ \hfill \textit{\small Update red vertices}
  \end{itemize}
  Last, $D_3\leftarrow  \bigcup\limits_{0\le i\le d} \pick_i$.
\end{tcolorbox}

\smallskip
Let us first prove that the algorithm in fact computes a dominating set.
\begin{lemma}\label{lem:correctness-general}
  When the algorithm has finished the iteration with parameter
  $i\geq 1$, then all vertices have residual degree at most $i-1$.
\end{lemma}

In particular, after finishing the iteration with parameter $1$, there
is no vertex with residual degree $1$ left and in the final round all
non-dominated vertices choose themselves into the dominating
set. Hence, the algorithm computes a dominating set of $G$.

\begin{proof}
  By induction, before the iteration with parameter $i$, all vertices
  have residual at most $i$. Assume $v$ has residual degree $i$ before
  the iteration with parameter $i$.  In that iteration, all
  non-dominated neighbors of $v$ choose a dominator (possibly $v$, then
  the statement is trivial),
  hence, are removed from $R$. It follows that the residual degree of $v$ after
  the iteration is $0$. Hence, after this iteration and before the
  iteration with parameter $i-1$, we are left with vertices of
  residual degree at most $i-1$.
\end{proof}

For the rest of this section analyze the size of $D_3$ and
we prove the following lemma.

\begin{lemma}\label{lem:greedy-approx}
 We have
 \[
 |D_3|\leq \left(\nabla_0\ln\Bigl(\frac{2{\dR}-4\nabla_0+1}{2\nabla_0+1}\Bigr)+3\nabla_0+1\right)\gamma_R.
 \]
\end{lemma}

Towards establishing the lemma we analyze the sizes of the sets
$\pick_i$ and $R_i$. The next
lemma follows from the fact that every vertex chooses at most one
dominator.

\begin{lemma}\label{lem:total-h}
  For every $0\leq i\le {\dR}$, $\sum\limits_{j\le i}|\pick_j| \le |R_i|$.
\end{lemma}
\begin{proof}
  The vertices of $R_i$ are those that remain to be dominated in the
  last $i$ rounds of the algorithm. As every vertex that remains to be
  dominated chooses at most one dominator in one of the rounds
  $j\leq i$, the statement follows.
\end{proof}

As the vertices of $D_R$ that still dominate non-dominated vertices also
have bounded residual degree, we can conclude that not too many
vertices remain to be dominated.
\begin{lemma}\label{lem:h1}
  For every $0\leq i\le {\dR}$, $|R_i| \le (i+1)\gamma_R$.
\end{lemma}
\begin{proof}
  First note that for every $i$,
  $D_R\setminus \bigcup_{j>i}\pick_j$ is a dominating
  set for $R_i$; additionally each vertex in this set has residual
  degree at most $i$. As every vertex dominates its residual neighbors and
  itself, we conclude $|R_i|\le (i+1)\gamma_R$.
\end{proof}

Finally, we show that we cannot pick too many vertices of high
residual degree. This follows from a simply density argument.
% bounded edge density.

\begin{lemma}\label{lem:h2}
  For every $2\nabla_0< i\leq {\dR}$, $|\pick_i|\leq \frac{\nabla_0}{i-2\nabla_0}(|R_i|-|R_{i-1}|)$.
\end{lemma}

\begin{proof}
Let $2\nabla_0< i\le {\dR}$ be an integer. We bound the size of $\pick_i$
  by a counting argument, using that $G$ (as well as each of its
  subgraphs) have edge density at most $\nabla_0$.

  Let $J := G[\pick_i]$ be the subgraph of $G$ induced by the
  vertices of $\pick_i$, which all have residual degree~$i$. Let
  $K := G[\pick_i \cup (N[\pick_i]\cap R_i)]$ be the subgraph of $G$
  induced by the vertices of $\pick_i$ together with the red
  neighbors that these vertices dominate.

  We have $|E(J)| \leq  \nabla_0|V(J)| = \nabla_0|\pick_i|$. As every
  vertex of $J$ has residual degree exactly $i$, we get
  $|E(K)| \geq i\pick_i - |E(J)| \geq (i-\nabla_0)|\pick_i|$ (we have to
  subtract $|E(J)|$ to not count twice the edges of $K$ that are
  between two vertices of $J$).  We also have
  $|V(K)| \le |V(J)| + |N[\pick_i]\cap R_i|$ and $|E(K)| \leq \nabla_0|V_K|$, hence
  $(i-2\nabla_0)|\pick_i| \leq \nabla_0|N[\pick_i]\cap R_i)|$.
  Now, as $R_{i-1} = R_i \setminus N[\pick_{i}]$, that is,
  $N[\pick_i]\cap R_i=R_i\setminus R_{i-1}$, we get
  $|\pick_i|\leq \frac{\nabla_0}{i-2\nabla_0}(|R_i|-|R_{i-1}|)$.
\end{proof}

% \pagebreak
Let $r_i\coloneqq |R_i|/\gamma_R$ and $d_i\coloneqq |\pick_i|/\gamma_R$.
Our goal is to maximize $S\coloneqq\sum_{0\leq i\leq {\dR}}d_i$ (which we have
to multiply by $\gamma_R$ in the end) under the constraints $d_i\geq 0$
and
\begin{align}
	r_i &\ge \sum_{j\le i} d_j&\text{($0\le i\le {\dR}$)}\label{eq:1}\\
	r_i &\le	i+1&\text{($0\le i\le {\dR}$)}\label{eq:2}\\
%	d_i&\leq \frac{\nabla_0}{i-2\nabla_0}r_{i}&\text{($2\nabla_0 +1\leq i\leq {\dR}$)}\label{eq:3}\\
	d_{i}&\leq \frac{\nabla_0 }{i-2\nabla_0 }\,(r_{i}-r_{i-1})&\text{($2\nabla_0 < i\le {\dR}$)}\label{eq:4}
\end{align}

%\sebi{rename $D$, this is reserved for the dominating set.}
%In the following for readability we write ${\dR}$ for ${\dR}$.
We may assume that ${\dR}\geq 3\nabla_0$, as otherwise,
\cref{lem:greedy-approx} follows immediately from \cref{lem:total-h}.

\smallskip
Let $a$ be the minimum integer such that $d_a>0$.

\begin{lemma}
	We can assume $r_i=0$ for all $i<a$.
\end{lemma}
\begin{proof}
	Putting $r_i=0$ for all $i<a$ obviously preserves \cref{eq:1,eq:2}. It also preserves \cref{eq:4} as the only case
	to check is $i=a-1$ (if $a\ge 2\nabla_0 $), for which the right hand side was possibly increased.
\end{proof}

\begin{lemma}
	If $a\le 3\nabla_0 -1$, then decreasing $d_a$ to $0$ and $r_a$ to $r_a-d_a$ and increasing $d_{a+1}$ to $d_a+d_{a+1}$ preserves  all the constraints and the value of $S$.
\end{lemma}
\begin{proof}
	The sum in \cref{eq:1} does not change if $i>a$ and \cref{eq:1} is obviously satisfied after modifications for $i\le a$.
	\cref{eq:2} is trivially satisfied after modification, as no $r_i$ increases.
	The only changes for \cref{eq:4} correspond to the case $i=a-1$ (for which the left hand side decreases,  while the right hand side increases) or to the case $i=a$ (for which the left hand side increases by $d_a$,  while the right hand side increases by $\nabla_0 /(a+1-2\nabla_0 )\,d_a\ge d_a$).
\end{proof}

From the above lemmas, as $r_a\geq d_a$, it follows that we may assume $a\ge 3\nabla_0 $ and $r_i=0$ for all $i<a$.

%Note that for $i=a-1$, \cref{eq:4} follows from \cref{eq:1}. Hence, if we modify the values, we only have to check \cref{eq:4} for $i\geq a$.

Note that \cref{eq:4} implies
\begin{equation}
	r_{2\nabla_0 }\le r_{2\nabla_0 +1}\le \dots\le r_{\dR}.
\end{equation}

Remark that increasing $r_{\dR}$ obviously preserves \cref{eq:1,eq:4}. Hence, we can assume that $r_{\dR}={\dR}+1$.
Let $b$ be minimum with $r_{i}=i+1$ for all $i\geq b$. Note that $b\ge a$.

\begin{lemma}
	Let $\alpha=\min(b-r_{b-1},\sum_{j<b-1}d_j)$.
If $b\ge 3\nabla_0 +1$, then increasing  $d_{b-1}$ and $r_{b-1}$ by $\alpha$ and decreasing
$\sum_{j<b-1}$ by $\alpha$ preserves the constraints and the value of $S$.
\end{lemma}
\begin{proof}
	\cref{eq:1,eq:2} are obviously preserved. For \cref{eq:4} we have to check the case where $i=b-1$ (for which the right hand side decreases by $\nabla_0 /(b-1-2\nabla_0 )\,\alpha\le\alpha$ and the left hand side decreases by $\alpha$) and the case $i=b-2$  (for which the right hand side increases and the left hand side decreases).
\end{proof}
Applying this lemma, either we can reduce $b$ to $3\nabla_0 $ (hence $b=a$), or we force $d_i=0$ for all $i<b-1$. Thus, $a=b-1$ or $a=b$.

\begin{lemma}
	We can assume that for every $b< i\le {\dR}$ we have  $d_i=\nabla_0 /(i-2\nabla_0 )$.
\end{lemma}
\begin{proof}
Indeed, as $b\ge a\ge 3\nabla_0 $, for $b< i\le {\dR}$,  \cref{eq:1,eq:2,eq:4} reduce to
$d_{i}\le \frac{\nabla_0 }{i-2\nabla_0 }$. Hence, we can assume $d_i=\nabla_0 /(i-2\nabla_0 )$ if $i> b$.
\end{proof}

\begin{lemma}
	We can assume $b=a$.
\end{lemma}
\begin{proof}
	Assume $a=b-1$ and let $\alpha=b-r_a$.
	We have $d_a\le\frac{\nabla_0 }{a-2\nabla_0 } (b-\alpha)$ and $d_b\le \frac{\nabla_0 }{b-2\nabla_0 } (1+\alpha)$.
	If we increase $r_a$ to $b$, we can increase $d_a$ by $\frac{\nabla_0 \alpha}{a-2\nabla_0 }$ and decrease
	$d_b$ by  $\frac{\nabla_0 \alpha}{b-2\nabla_0 }$, while preserving the constraints and increasing $S$.
\end{proof}

\begin{lemma}
	We can assume $a=3\nabla_0 $.
\end{lemma}
\begin{proof}
Assume $a\ge 3\nabla_0 +1$. By putting $r_{a-1}=a$ we can increase $d_{a-1}$ by $\frac{\nabla_0 }{a-1-2\nabla_0 }a$ and decrease~$d_a$ by  $\frac{\nabla_0 }{a-2\nabla_0 }a$. Note that the condition $a\ge 3\nabla_0 +1$ implies that
$\frac{\nabla_0 }{a-1-2\nabla_0 }\le 1$, which is needed to preserve \cref{eq:1}.
\end{proof}

Now we have $a=b=3\nabla_0 $ and we can put $d_a=a+1$.
Hence,  the optimum is

$d_i=0$ if $i<3\nabla_0 $, $d_{3\nabla_0 }=3\nabla_0 +1$ and $d_{3\nabla_0 +i}=\nabla_0 /(\nabla_0 +i)$.
Altogether, we get
\[
S=3\nabla_0 +1+\nabla_0 \sum_{i=\nabla_0 +1}^{{\dR}-2\nabla_0 }\frac{1}{i}=\nabla_0 (H_{{\dR}-2\nabla_0 }-H_{\nabla_0} )+3\nabla_0 +1,
\]
where $H_i=1+1/2+\dots+1/i$ is the $i$th harmonic number.

It is known \cite{DeTemple1991} that
\[
\frac{1}{24(n+1)^2}<H_n-\ln \Bigl(n+\frac12\Bigl)-\gamma<\frac{1}{24n^2},
\]
where $\gamma$ is the Euler--Mascheroni constant.
We deduce that for $n>m$ we have
\[
-\frac{1}{24m^2}<\frac{1}{24}\Bigl(\frac{1}{(n+1)^2}-\frac{1}{m^2}\Bigr)<(H_n-H_m)-\ln \Bigl(\frac{2n+1}{2m+1}\Bigr)<\frac{1}{24}\Bigl(\frac{1}{n^2}-\frac{1}{(m+1)^2}\Bigr)\le 0
\]

 we deduce
 \[
 -\frac{1}{24\nabla_0 ^2}<S-\biggl(
 \nabla_0 \ln\Bigl(\frac{2{\dR}-4\nabla_0 +1}{2\nabla_0 +1}\Bigr)+3\nabla_0 +1\biggr)<0.
 \]

 Hence, with a (negative) error less than $0.042$ we have
 \[
 S\approx \nabla_0 \ln\Bigl(\frac{2{\dR}-4\nabla_0 +1}{2\nabla_0 +1}\Bigr)+3\nabla_0 +1.
 \]

 This finishes the proof of \cref{lem:greedy-approx}. \hfill$\qed$

% \bigskip
%
%Assume $\nabla_0>1$.
%Choose $s=\nabla_0+1$, $t=\nabla_0(\nabla_0+1)$.
%Then, $2\nabla_0+(\nabla_0+1)^2\le (2\nabla_0)^2\le \kappa^2$.
%Thus,
%
%\[
%{\dR}=\kappa^{s-1}(t+s-1+2(s-1)\kappa^3)=
%\kappa^{\nabla_0+3}(2\nabla_0+1)\le 2\kappa^{\nabla_0+4}
%\]
%Hence,
%\[
%\ln((2{\dR}-4\nabla_0 +1)/(2\nabla_0 +1))<\ln {\dR}\le (\nabla_0+4)\ln\kappa+1.
%\]
%
%From which follows
%\[
%|D_3|<\bigl((\nabla_0+4)\ln(2\nabla_1)+1\bigr)\gamma_R.
%\]
%
%Note that if we require $\nabla_0>\nabla_0(G)$, then we have
%\[
%|D_3|<\bigl(\nabla_0(\nabla_0+3)\ln(2\nabla_1)+4\nabla_0+1\bigr)\gamma_R.
%\]

%With the general bound for ${\dR}$, we get (for $\nabla_1\ge 1$)
%\[
%{\dR}<(2\nabla_1)^{2\nabla_1}(4\nabla_1+1+128\nabla_1^5)<(2\nabla_1)^{2\nabla_1}(e^5\nabla_1^5).
%\]
%Hence,
%\[
%\ln((2{\dR}-4\nabla_0 +1)/(2\nabla_0 +1))<\ln {\dR}<(2\nabla_1)\ln(2\nabla_1)+5\nabla_1+5.
%\]
%Thus, as $\nabla_0 \le\nabla_1$ we have
%\[
%S< 2\nabla_1^2\ln (2\nabla_1)+5\nabla_1^2+8\nabla_1+1
%\]
%In particular, for $K_{t,t}$-free graphs, $S=O((t\ln t)^2)$.

% !TEX root = main.tex

\section{$K_{3,t}$-free graphs}

We now turn our attention to graphs 
that exclude $K_{3,t}$ for some $t$ (and with bounded $\nabla_1$). 
The most prominent graphs
with these properties are graphs that embed into a surface of bounded
genus and in particular planar graphs.

\subsection{Phase 1: Preprocessing}

The general preprocessing phase described in \cref{sec:step1} remains 
unchanged. Recall that we defined $D_1$ as $\{v\in V(G) : \text{ for all } A\subseteq V(G)\setminus \{v\} \text{ with $N(v)\subseteq N[A]$ we have $|A|> (2\nabla-1)\}$}.$ Recall as well that for every $v\in V(G)\setminus D_1$, we fixed
  $A_v\subseteq V(G)\setminus \{v\}$ such that
  $\Nr(v)\subseteq N[A_v]$ and $|A_v|\leq 2\nabla-1$. Furthermore, 
  for $V(G)\setminus \hat{D}$ we
assumed $A_v\subseteq D\setminus\{v\}$.
 
\subsection{Phase 2: local dominators in the $K_{3,t}$-free case}\label{sec:D2}

In this second phase, things get simpler than in \cref{sec:phase2}. Since
we now assume that we exclude $K_{3,t}$ the domination sequences of \cref{def:dom-sequence} only have
length two. We can therefore simplify the analysis of the domination sequences. We simply select every pair
of vertices with sufficiently many neighbors in common.

\begin{tcolorbox}[colback=red!5!white,colframe=red!50!black]
\begin{itemize}
\item For $v\in V(G)$ let
  $B_v\coloneqq \{z\in V(G)\setminus \{v\}: |\Nr(v)\cap \Nr(z)|\geq
  (2\nabla-1)t+1\}$.\smallskip
\item Let $W$ be the set of vertices $v\in V(G)$ such that
  $B_v \neq \emptyset$.\smallskip
\item Let $D_2\coloneqq \bigcup\limits_{v\in W} (\{v\}\cup B_v)$.
\end{itemize}
\end{tcolorbox}

Once $D_1$ has been computed in the previous phase, 2 more rounds of
communication are enough to compute the sets $B_v$ and $D_2$.
Before we update the residual degrees, let us analyze the sets $B_v$
and~$D_2$.  First note that the definition is symmetric: since
$\Nr(v)\cap \Nr(z)=\Nr(z)\cap \Nr(v)$ we have for all $v,z\in V(G)$ if
$z\in B_v$, then $v\in B_z$. In particular, if $v\in D_1$ or
$z\in D_1$, then $\Nr(v)\cap \Nr(z)=\emptyset$, which immediately
implies the following lemma.

\begin{lemma}\label{lem:WcapD1}
  We have $W\cap D_1=\emptyset$ and for every $v\in V(G)$ we have
  $B_v\cap D_1=\emptyset$.
\end{lemma}
%\begin{proof}
%  Let $v,z\in V(G)$. If $v\in D_1$, then
%  $\Nr(v)=\emptyset$. Similarly if $z\in D_1$, then
%  $\Nr(v)\cap N(z)=N(v)\cap \Nr(z) =\emptyset$.
%\end{proof}

Now we prove that for every $v\in W$, the set $B_v$ cannot be too big,
and has nice properties. 

\begin{lemma}\label{lem:dominating-dominators}
  For all vertices $v\in W$ we have

  \vspace{-5pt}
  \begin{itemize}
  \item $B_v \subseteq A_v$, (hence $|B_v|\leq (2\nabla-1)$) and \smallskip
  \item if $v\not\in \hat{D}$, then $B_v\subseteq D$.
  \end{itemize}
\end{lemma}

\begin{proof}
  Assume $A_v=\{v_1,\ldots, v_\ell\}$ (a set of possibly not distinct
  vertices) and assume there exists
  $z\in V(G)\setminus \{v,v_1,\ldots v_\ell\}$ with
  $|\Nr(v) \cap \Nr(z)| \geq (2\nabla-1)t+1$.  As $v_1, \ldots, v_\ell$ dominate~$\Nr(v)$,
  and hence also \mbox{$\Nr(v)\cap \Nr(z)$}, and $\ell\leq (2\nabla-1)$, there must be some~$v_i$,
  $1\leq i\leq \ell$, with
  \mbox{$|\Nr(v) \cap \Nr(z) \cap N[v_i]| \geq \lceil ((2\nabla-1)t+1)/(2\nabla-1)\rceil \geq
    t+1$}.  Therefore, $|\Nr(v) \cap \Nr(z) \cap N(v_i)| \geq t$,
  which shows that $K_{3,t}$ is a subgraph of~$G$, contradicting the
  assumption.

  If furthermore $v\not\in \hat{D}$, by definition of $\hat{D}$, we
  can find $w_1,\ldots, w_\ell$ from $D$ that dominate $N(v)$, and in
  particular $\Nr(v)$.  If $z\in V(G)\setminus \{v,w_1,\ldots, w_\ell\}$
  with $|\Nr(v) \cap \Nr(z)| \geq (2\nabla-1)t+1$ we can argue as above to obtain
  a contradiction.
\end{proof}

\pagebreak

% In the light of \cref{lem:dominating-dominators}, we select all
% paires of nodes with sufficiently large intersecting neighborhood.
%
% \begin{tcolorbox}
%   For $v\in V(G)$ let $B_v\coloneqq \{z\in V(G)\setminus \{v\}:
%   |N(v)\cap N(z)|\geq 19\}$.
% \end{tcolorbox}

% \begin{corollary}\label{cor:dominating-dominators}
%   For every vertex $v$, $B_v\subseteq A_v$, in particular,
%  $|B_v|\leq 6$ and if $v\not\in \hat{D}$, then $B_v\subseteq D$.
% \end{corollary}
%
% \begin{tcolorbox}
%   We define $W$ as the set of vertices $v\in V(G)$ such
%   that $B_v\neq \emptyset$. We define \[D_2\coloneqq \bigcup_{v\in W}
%   (\{v\}\cup B_v).\]
% \end{tcolorbox}
% \vspace{0mm}

% Our algorithm now proceeds as follows. Obviously, every vertex $v$
% can locally compute the set $B_v$. The algorithm adds the set $D_2$
% to the dominating set, removes~$D_2$ from the graph and marks all
% vertices dominated by $D_2$ as dominated.  \alex{here again 'remove'
% vertices?}

Let us now analyze the size of $D_2$. For this we refine the set $D_2$
and define
\begin{tcolorbox}
  \begin{enumerate}
    \item $D_2^1\coloneqq \bigcup_{v\in W\cap D}
    (\{v\}\cup B_v)$, \smallskip
    \item $D_2^2\coloneqq \bigcup_{v\in W\cap (\hat{D}\setminus D)}
    (\{v\}\cup B_v)$, and \smallskip
    \item $D_2^3\coloneqq \bigcup_{v\in W\setminus (D\cup \hat{D})}
    (\{v\}\cup B_v)$.
  \end{enumerate}
\end{tcolorbox}

\smallskip
Obviously $D_2=D_2^1\cup D_2^2\cup D_2^3$. We now bound the size of the
refined sets $D_2^1,D_2^2$ and $D_2^3$.

\begin{lemma}\label{lem:size-D21}
  $|D_2^1\setminus D|\leq (2\nabla-1)\gamma$.
\end{lemma}
\begin{proof}
  We have
  \[|D_2^1\setminus D|= |\bigcup_{v\in W\cap D} (\{v\}\cup
    B_v)\setminus D|\leq |\bigcup_{v\in W\cap D}B_v|\leq \sum_{v\in
      W\cap D}|B_v|.\] By \cref{lem:dominating-dominators} we have
  $|B_v|\leq (2\nabla-1)$ for all $v\in W$ and as we sum over $v\in W\cap D$ we
  conclude that the last term has order at most $(2\nabla-1)\gamma$.
\end{proof}

\begin{lemma}\label{lem:size-D22}
  $D_2^2 \subseteq \hat D$ and therefore
  $|D_2^2\setminus D|< \rho(G)\gamma$.
\end{lemma}
\begin{proof}
  Let $v\in \hat{D}\setminus D$ and let $z\in B_v$. By symmetry,
  $v\in B_z$ and according to \cref{lem:dominating-dominators}, if
  $z\not\in \hat{D}$, then $v\in D$.  Since this is not the case, we
  conclude that $z\in\hat{D}$.  Hence $B_v\subseteq \hat{D}$ and, more
  generally, $D_2^2\subseteq \hat{D}$.  Finally, according to
  \cref{lenzen-improved} we have $|\hat{D}\setminus D|<\rho(G)\gamma$.
\end{proof}

Finally, the set $D_2^3$, which appears largest at first glance, was
actually already counted, as shown in the next lemma.
\begin{lemma}\label{lem:size-D23}
  $D_2^3\subseteq D_2^1$.
\end{lemma}
\begin{proof}
  If $v\not\in \hat{D}$, then $B_v\subseteq D$ by
  \cref{lem:dominating-dominators}.  Hence $v\in B_z$ for some
  $z\in D$, and $v\in D_2^1$.
\end{proof}

Recall that we defined $\eta\in [0,1]$ to be such that $|(D_1 \cup D_2)\cap D| =\eta\gamma$.

% \begin{lemma}\label{lem:size-D2}
%   We have that $|D_2| < 3\gamma + 7\e\gamma$.
% \end{lemma}
% \begin{proof}
%   First, by \cref{lem:size-D22}, we have $|D_2^2|< 3\gamma+\e\gamma$. Then, with
%   \cref{lem:size-D21}, we have $|D_2^1|< 6\e\gamma$. Finally, with
%   \cref{lem:size-D23} we conclude that $|D_2|<3\gamma + 7\e\gamma$
% \end{proof}
%
% \begin{lemma}\label{lem:size-D1}
%   We have that $|D_1| < 3\gamma + \e\gamma$.
% \end{lemma}

% \alex{replacement for Lemmas 7 and 8 below}

\smallskip
\begin{lemma}\label{lem:size-D12}
  We have $|D_1\cup D_2| < \rho(G)\gamma+2\nabla\eta\gamma$.
\end{lemma}
\begin{proof}
  By \cref{lem:size-D23} we have $D_2^3\subseteq D_2^1$, hence,
  $D_1\cup D_2=D_1\cup D_2^1\cup D_2^2$. By \cref{lem:size-D1} we have
  $D_1 \subseteq \hat D$ and by \cref{lem:size-D22} we also have
  $D_2^2 \subseteq \hat D$, hence $D_1\cup D_2^2\subseteq \hat D$.
  Again by \cref{lenzen-improved}, $|\hat D \setminus D|<\rho(G)\gamma$ and
  therefore $|(D_1 \cup D_2^2 )\setminus D| < \rho(G) \gamma$.

  We have $W\cap D\subseteq D_2^1\cap D$, hence with
  \cref{lem:dominating-dominators} we conclude that
  \[
    \big\vert D_2^1 \setminus D \big\vert \leq
    \Big\vert\bigcup\limits_{v\in D \cap D_2^1}B_v\Big\vert \leq
    \sum\limits_{v\in D \cap D_2^1} |B_v| \leq (2\nabla-1)\rho\gamma,
  \]
  hence $(D_1\cup D_2)\setminus D<\rho(G)\gamma+(2\nabla-1)\eta\gamma$. Finally,
  $D_1\cup D_2=(D_1\cup D_2)\setminus D\cup ((D_1\cup D_2)\cap D)$ and
  with the definition of $\eta$ we conclude
  $|D_1\cup D_2|<\rho(G)\gamma + 2\nabla\eta\gamma$.
\end{proof}
%The analysis of the next and final step of the algorithm will actually
%show that the worst case is obtained when $\eta=0$.

We now update the residual degrees, that is, we update $R$ as
$V(G)\setminus N[D_1\cup D_2]$ and for every vertex the number
$\dr(v)=|\Nr(v)|$ accordingly.

Just as before, we show that after the first two phases of the algorithm we
are in the very nice situation where all residual degrees are
small. 

\begin{lemma}\label{lem:res-degree}
  For all $v\in V(G)$ we have $\dr(v)\leq (2\nabla-1)^2t+(2\nabla-1)$.
\end{lemma}
\begin{proof}
  First, every vertex of $D_1\cup D_2$ has residual degree $0$.
  Assume that there is a vertex $v$ of residual degree at least $(2\nabla-1)^2t+(2\nabla-1)+1$.
  As $v$ is not in $D_1$, its residual neighbors are
  dominated by a set $A_v$ of at most $(2\nabla-1)$ vertices. Hence there is a
  vertex~$z$ (not in $D_1$ nor $D_2$) with $|\Nr(v)\cap \Nr[z]|\geq (2\nabla-1)t+2 = ((2\nabla-1)t+1)+1$, hence, $|\Nr(v)\cap \Nr(z)|\geq (2\nabla-1)t+1$.
  This contradicts that $v$ is not in~$D_2$.
\end{proof}

% !TEX root = main.tex

\subsection{Phase 3: LP-based approximation}\label{sec:LP-planar}
%
%\sebi{Update the general case such that we can work
%with arbitrary $\epsilon$.}
%
%Recall that we chose $\rho$ such that $|(D_1 \cup D_2)\cap D| =\rho\gamma$.
%When $\rho\gamma$ vertices of $D$ were already chosen into
%the partial dominating set $D_1\cup D_2$ we have $|D_R|\leq (1-\rho)\gamma$.
%With \cref{cor:planar-orientations} we conclude the following corollary.
%
%\sebi{Put this to the general case.}
%
%\begin{corollary}\label{cor:LP-approx}
%  Let $G$ be a graph that has an orientation with maximum out-degree
%  $d$, let $R\subseteq V(G)$, let~$D_R$ be a
%  minimum dominating set of $R$, and let $\e>0$. Then we can
%  compute a set $D'$ of size at most~$(2d+1)(1+\e)|D_R|$ that dominates
%  $R$ in $\Oof(\log \Delta/(\log (1+\e))$ rounds in the LOCAL
%  model.
%
%  In particular, for our algorithm when
%
%  \vspace{-2mm}
%  \begin{enumerate}
%    \item $G$ is planar, then $|D'|\leq 7(1+\e)|D_R|
%      \leq 7(1+\e)(1-\rho)\gamma$, and when
%    \item $G$ is planar and triangle-free or outerplanar, then
%      $|D'|\leq 5(1+\e)|D_R|\leq 5(1+\e)(1-\rho)\gamma$.
%  \end{enumerate}
%\end{corollary}

We now proceed with the LP-based approximation as in the general case
presented in \cref{sec:LP}. Recall that for any desired $\epsilon>0$
we defined $\Gamma$ as an high degree and defined
$D_3^1$ as the set of all vertices with degree greater than $\Gamma$.
We added $D_3^1$ to the dominating set and were able to call the
LP-based approximation algorithm of \cref{cor:LP-approx-general}. We finally
obtained a set $D_3$ dominating the remaining vertices with
$|D_3|\leq (2\nabla_0+1)(1+\epsilon)\gamma_R$ according to \cref{lem:D3-LP}.

We now conclude our main theorem, \cref{thm:K3t-free-total}, stating
that the algorithm on $K_{3,t}$-free graphs computes a dominating
set of size at most $(6\nabla_1+3)\gamma$.

\begin{proof}[Proof of \cref{thm:K3t-free-total}]
First, $D_1,D_2$, and $D_3$ are computed locally, in a bounded number of
  rounds, and additionally the set $D_1 \cup D_2 \cup D_3$ dominates $G$. Then,
  \begin{align*}
  	|D_1\cup D_2|&<\rho(G)\gamma+2\nabla\eta\gamma&\text{(by \cref{lem:size-D12})}\\
  	&\leq (2\nabla_1+1)\gamma+2\nabla_1\eta\gamma&\text{(as $\rho(G)\leq 2\nabla_1+1$ and $\nabla\leq \nabla_1$)}
	\intertext{and}
  	|D_3|&\leq (2\nabla_0+1)(1+\epsilon)(1-\eta)\gamma&\text{(by \cref{lem:D3-LP})}\\
	&\leq (2\nabla_1+1)(1+\epsilon)(1-\eta)\gamma&\text{(as $\nabla_0\leq \nabla_1$)}
  	\intertext{ By choosing $\epsilon=1$,}
  	|D_1 \cup D_2 \cup D_3|  &\leq (2\nabla_1+1+2\nabla_1\eta+(4\nabla_1+2)(1-\eta))\gamma\\
  	&\leq  (6\nabla_1+3)\gamma&\text{(maximized when $\eta=0$)}
  \end{align*}
\end{proof}

\subsection{Planar graphs}

Finally, we complete the analysis of our algorithm for planar graphs,
showing that it computes an $(11+\epsilon)$-approximation. In the
following we fix a planar graph $G$.

\begin{proof}[Proof of \cref{thm:planar}]
We revisit the proof of \cref{thm:K3t-free-total} and plug in the
numbers for planar graphs. For planar graphs we have
$K_{3,3}\not\subseteq G$, $\nabla_0=3$, $\nabla=2$ and
$\rho=4$, as stated in \cref{lem:bounds}.
Therefore, by \cref{lem:size-D12} we have $|D_1\cup D_2|<4\gamma+4\eta\gamma$   and by \cref{lem:D3-LP} we have \linebreak $|D_3|\leq (7+\e)(1-\eta)\gamma$. Hence, $|D_1 \cup D_2 \cup D_3|\leq \gamma(4+4\eta +7 - 7\eta +\e - \e\eta)
  \leq \gamma(11+\e-3\eta -\e\eta)$.
  As $\eta\in[0,1]$, this is maximized when $\eta=0$. Hence
  $|D_1 \cup D_2 \cup D_3| \le \gamma(11+\e)$.
\end{proof}

%\input{9-residual-degrees}

% !TEX root = main.tex

%\subsection{Restricted classes of planar graphs}
We now further restrict the input graphs, requiring e.g.\ planarity
together with a lower bound on the girth. Our algorithm works exactly as before, however,
using different parameters. From the different edge densities and
Hall ratio of
the restricted graphs we will then derive different constants and
as a result a better approximation factor. Throughout this section
we use the same notation as in the first part of the paper.

As in the general case in the first phase we begin by computing
the set $D_1$ and analyzing it in terms of the auxiliary set $\hat{D}$.

%The proofs of most lemmas
%can be easily adapted from the general case and we will only go into
%detail when needed. In \cref{sec:triangle-free} we prove that on
%triangle-free planar graphs we can compute a 16-approximation, improving
%the currently best known bound of 32. In \cref{sec:girth} we prove
%that on planar graphs of girth 5 we can compute a 9-approximation,
%improving the best known bound of 18, and finally in \cref{sec:outer}
%we prove that on outerplanar graphs our algorithm computes a 13-approximation. This does not improve the tight known bound
%of 5, however, it demonstrates that our algorithm works robustly
%on this interesting subclass.

\begin{corollary}\label{a-lenzen-improved-planar}
  ~
  \begin{enumerate}
    \item If $G$ is bipartite, then $|\hat{D}\setminus D| < 2\gamma$.\smallskip
    \item If $G$ is triangle-free, outerplanar, or has girth 5,
      then $|\hat{D}\setminus D| < 3\gamma$.
  \end{enumerate}
\end{corollary}
\begin{proof}
This is immediate from \cref{lem:bounds} and \cref{lenzen-improved}.
\end{proof}

The inclusion $D_1\subseteq \hat D$ continues to hold and the bound
on the sizes as stated in the next lemma is again a direct consequence of the corollary.

\smallskip
\begin{lemma}\label{alem:D-hat}
    We have $D_1\subseteq \hat D$, and
  \textit{\begin{enumerate}
    \item if $G$ is bipartite,
      then $|\hat{D}\setminus D| < 2\gamma$ and $|\hat{D}|< 3\gamma$.\smallskip
    \item if $G$ is triangle-free, outerplanar, or has girth 5,
      then $|\hat{D}\setminus D| < 3\gamma$ and $|\hat{D}|< 4\gamma$.
  \end{enumerate}}
\end{lemma}

In case of triangle-free planar graphs (in particular in the case of bipartite
planar graphs) we proceed with the second phase exactly as in the second phase of
the general algorithm (\cref{sec:D2}), however, the parameter $(2\nabla-1)t+1$ is replaced by
the parameter $7$. In  case of planar graphs of girth at least five or outerplanar
graphs, we simply set $D_2=\emptyset$.

\begin{tcolorbox}[colback=red!5!white,colframe=red!50!black]
  If $G$ is triangle-free:

  \begin{itemize}
    \item For $v\in V(G)$ let $B_v\coloneqq \{z\in V(G)\setminus
      \{v\}: |N_R(v)\cap N_R(z)|\geq 7\}$.\smallskip
    \item Let $W$ be the set of vertices $v\in V(G)$ such
      that $B_v \neq \emptyset$.\smallskip
    \item Let $D_2\coloneqq \bigcup\limits_{v\in W} (\{v\}\cup B_v)$.
  \end{itemize}

  If $G$ has girth at least $5$ or $G$ is outerplanar, let $D_2=\emptyset$.
\end{tcolorbox}

\cref{lem:WcapD1} is based only on the definition of $B_v$ and $W$ and
does not use particular properties of planar graphs, hence, it also holds
in the restricted case.

The next lemma uses the triangle-free property.

\begin{lemma}\label{alem:dominating-dominators}
  If $G$ is triangle-free, then for all vertices $v\in W$ we have

  \vspace{-5pt}
  \begin{itemize}
  \item $B_v \subseteq A_v$ (hence $|B_v|\le 3$), and \smallskip
  \item if $v\not\in \hat{D}$, then $B_v\subseteq D$.
  \end{itemize}
\end{lemma}
\begin{proof}
  Assume $A_v=\{v_1,v_2, v_3\}$ and assume there is $z\in V(G)\setminus \{v,v_1,v_2, v_3\}$
  with $|\Nr(v) \cap \Nr(z)| \geq 7$.
  As the vertices $v_1, v_2, v_3$ dominate $\Nr(v)$, and hence $\Nr(v)\cap \Nr(z)$,
  there must be some~$v_i$, $1\leq i\leq 3$, with
  \mbox{$|\Nr(v) \cap \Nr(z) \cap N[v_i]| \geq \lceil 7/3\rceil \geq 3$}.
  Then on of the following holds: either
  \mbox{$|\Nr(v) \cap \Nr(z) \cap N(v_i)| \geq 3$},  or
  \mbox{$|\Nr(v) \cap \Nr(z) \cap N(v_i)| =2$}.
  The first case shows that $K_{3,3}$ is a subgraph of~$G$
  contradicting the assumption that $G$ is planar.
  The second case implies that $v_i\in \Nr(v)$. In this situation, by picking
  $w \in \Nr(v) \cap \Nr(z) \cap N(v_i)$, we get that $(v,v_i,w)$ is a triangle,
  hence we also reach a contradiction.

  If furthermore $v\not\in \hat{D}$, by definition of $\hat{D}$,
  we can find $w_1,w_2, w_3$ from $D$
  that dominate $N(v)$, and in particular $\Nr(v)$.
  If $z\in V(G)\setminus \{v,w_1,w_2, w_3\}$
  with $|\Nr(v) \cap \Nr(z)| \geq 7$ we can argue as above to obtain
  a contradiction.
\end{proof}

For our analysis we again split $D_2$ into three sets $D_2^1, D_2^2$ and
$D_2^3$. The next lemmas hold also for the restricted cases. We repeat
them for convenience with the appropriate numbers filled it.

\begin{lemma}\label{alem:size-D21}
  If $G$ is triangle-free, then $|D_2^1\setminus D|\leq 3\gamma$.
\end{lemma}

\begin{lemma}\label{alem:size-D22}
  If $G$ is triangle-free, then $D_2^2 \subseteq \hat D$ and therefore
  $|D_2^2\setminus D|< 3\gamma$.
\end{lemma}

\begin{lemma}\label{alem:size-D23}
  If $G$ is triangle-free, then $D_2^3\subseteq D_2^1$.
\end{lemma}

Recall that $\eta\in [0,1]$ is such that $|(D_1 \cup D_2)\cap D| =\eta\gamma$.

\begin{lemma}\label{alem:size-D12}
\mbox{ }
\vspace{-1mm}
\begin{enumerate}
\item If $G$ is bipartite, then $|D_1\cup D_2| < 2\gamma+4\eta\gamma$.
\item If $G$ is triangle-free, then $|D_1\cup D_2| < 3\gamma+4\eta\gamma$.
\item If $G$ has girth at least $5$ or is outerplanar, then $|D_1\cup D_2| < 3\gamma+\eta\gamma$.
\end{enumerate}
\end{lemma}
\begin{proof}
If $G$ is outerplanar or $G$ has girth at least $5$, then $D_2=\emptyset$.
By \cref{alem:D-hat} we have $D_1\subseteq \hat{D}$ and
$|\hat{D}\setminus D|<3\gamma$, hence $(D_1\cup D_2)\setminus D<3\gamma$.

If $G$ is triangle-free, by \cref{alem:size-D23} we have $D_2^3\subseteq D_2^1$, hence,
  $D_1\cup D_2=D_1\cup D_2^1\cup D_2^2$. By \cref{alem:D-hat} we have
  $D_1 \subseteq \hat D$ and by \cref{alem:size-D22} we also have
  $D_2^2 \subseteq \hat D$, hence $D_1\cup D_2^2\subseteq \hat D$.
  Again by \cref{alem:D-hat}, if $G$ is bipartite, then
  $|\hat D \setminus D|<2\gamma$, therefore $|(D_1 \cup D_2^2 )\setminus D| < 2\gamma$, and if~$G$ is triangle-free,
  then $|\hat D \setminus D|<3\gamma$,
  therefore $|(D_1 \cup D_2^2 )\setminus D| < 3\gamma$.
  We have $W\cap D\subseteq D_2^1\cap D$, hence with
  \cref{alem:dominating-dominators} we conclude that
  \[
    \big\vert D_2^1 \setminus D \big\vert \leq
    \Big\vert\bigcup\limits_{v\in D \cap D_2^1}B_v\Big\vert \leq
    \sum\limits_{v\in D \cap D_2^1} |B_v| \leq 3\eta\gamma,
  \]
  hence $(D_1\cup D_2)\setminus D<2\gamma+3\eta\gamma$ if
  $G$ is bipartite and $(D_1\cup D_2)\setminus D<3\gamma+3\eta\gamma$
  if $G$ is triangle-free.

  Finally,
  $D_1\cup D_2=(D_1\cup D_2)\setminus D\cup (D_1\cup D_2)\cap D$ and
  with the definition of $\eta$ we conclude

  \vspace{-2mm}
  \begin{enumerate}
  \item $|D_1\cup D_2|<2\gamma + 4\eta\gamma$ if $G$ is bipartite, 
  \item $|D_1\cup D_2|<3\gamma + 4\eta\gamma$ if $G$ is triangle-free, 
  \item $|D_1\cup D_2|<3\gamma + \eta\gamma$ if $G$ has girth at least
  $5$ or is outerplanar.
  \end{enumerate}

\end{proof}

Again, we update the residual degrees, that is, we update
$R$ as $V(G)\setminus N[D_1\cup D_2]$ and for every vertex the
number $\dr(v)=N(v)\cap R$ accordingly and proceed with
the third phase.

\begin{lemma}\label{alem:res-degree}
\textit{\begin{enumerate}
\item If $G$ is triangle-free, then $\Dr(v)\leq 18$.
\item If $G$ has girth at least $5$, then $\Dr(v)\leq 3$.
\item If $G$ is outerplanar, then $\Dr(v)\leq 9$.
\end{enumerate}}
\end{lemma}
\begin{proof}
  Every vertex of $D_1\cup D_2$ has residual degree $0$, hence, we
  need to consider only vertices that are not in $D_1$ or $D_2$.

  First assume that the graph is triangle-free and
  assume that there is a vertex $v$ of residual degree at least $19$.
  As $v$ is not in $D_1$, its $19$ non-dominated
  neighbors are dominated by a
  set $A_v$ of at most~3 vertices. Hence, there is vertex $z$ (not in $D_1$
  nor $D_2$) dominating at least $\lceil 19/3\rceil = 7$ of them.
  Here, $z$ cannot be one of these 7 vertices, otherwise it would be connected
  to $v$ and there would be a triangle in the graph.
  Therefore we
  have $|\Nr(v)\cap \Nr(z)|\geq 7$, contradicting that $v$ is not in~$D_2$.

  Now assume that $G$ has girth at least $5$ and
  assume that there is a vertex $v$ of residual degree at least~$4$.
  As $v$ is not in $D_1$, its $4$ non-dominated
  neighbors are dominated by a
  set $A_v$ of at most~3 vertices. Hence, there is vertex $z$ (not in $D_1$
  nor $D_2$) dominating at least $\lceil 4/3\rceil = 2$ of them.
  Here, $z$ cannot be one of these 2 vertices, otherwise it would be connected
  to $v$ and there would be a triangle in the graph. However, $z$ can
  also not be any other vertex, as otherwise we find a cycle of length $4$,
  contradicting that~$G$ has girth at least $5$.

  Finally, assume that $G$ is outerplanar and
  assume that there is a vertex $v$ of residual degree at least $10$.
  As $v$ is not in $D_1$, its $10$ non-dominated
  neighbors are dominated by a
  set $A_v$ of at most~3 vertices. Hence, there is vertex $z$ (not in $D_1$
  nor $D_2$) dominating at least $\lceil 10/3\rceil = 4$ of them.
  Therefore $|N(v)\cap N(z)|\geq 3$, and we find a $K_{2,3}$ as a
  subgraph, contradicting that $G$ is outerplanar.
\end{proof}

We proceed with the LP-based approximation as in the general case
presented in \cref{sec:LP}. Recall that for any desired $\epsilon>0$ 
we defined $\Gamma$ as an high degree and defined 
$D_3^1$ as the set of all vertices with degree greater than $\Gamma$. 
We added $D_3^1$ to the dominating set and were able to call the
LP-based approximation algorithm of \cref{cor:LP-approx-general}. We finally
obtained a set $D_3$ dominating the remaining vertices with 
$|D_3|\leq (2\nabla_0+1)(1+\epsilon/5)\gamma_R$ according to \cref{lem:D3-LP} (by applying it with $\epsilon'=\epsilon/5$). The second item does not follow by the LP approximation, but simply from \cref{alem:res-degree} and \cref{cor:size-R}. 

%
%We proceed to compute a dominating set of the remaining vertices
%as before for the respective number of rounds.
%
%\begin{lemma}\label{alem:total-H}
%  If $G$ is triangle-free or outerplanar,
%  for every $1\le i$, $\sum\limits_{j\le i}|\Delta_j| \le |R_i|$.
%\end{lemma}
%
%\begin{lemma}\label{alem:tri-h1}
%  If $G$ is triangle-free or outerplanar,
%  for every $1\le i$, $|R_i| \le (i+1)(1-\e)\gamma$.
%\end{lemma}
%
%\begin{lemma}\label{alem:tri-delta}
%  If $G$ is triangle-free or outerplanar,
%  for every $5\le i$, $|\Delta_i| \le \frac{2|R_i|}{i-4}$.
%\end{lemma}
%
%\begin{lemma}\label{alem:tri-h2}
%  If $G$ is triangle-free or outerplanar,
%  for every $1\le i$, $|R_i| \le |R_{i+1}| - \frac{(i-3)|\Delta_{i+1}|}{2}$.
%\end{lemma}
%
%The proofs of \cref{alem:total-H} to \ref{alem:tri-h2} are copies of the ones
%for \cref{lem:total-h,lem:h1,lem:delta,lem:h2}, with the execption that the edge
%density of $3$ for planar graphs if now replaced by $2$ for triangle-free and
%outerplanar.
%%
%Similarly to \cref{lem:size-D3} we formulate (and present in
%\cref{sec:linear-prog}) a linear program to maximize $|D_3|$ under these
%constraints, yielding the following lemma.

\begin{lemma}\label{alem:size-D3}

\textit{\begin{enumerate}
\item If $G$ is triangle-free, then  $|D_3|\le 5(1+\epsilon/5)\gamma_R$.\smallskip
\item If $G$ has girth at least $5$, then $|D_3|\le 4\gamma_R$.\smallskip
\item If $G$ is outerplanar, then $|D_3|\le 5(1+\epsilon/5)\gamma_R$.
\end{enumerate}}
\end{lemma}

\pagebreak
We conclude. 

\begin{theorem}\label{thm:tri}
There exists a distributed LOCAL algorithm that, for every triangle free planar
graph $G$, computes in a constant number of rounds a dominating set
of size at most $(8+\epsilon)\gamma(G)$.
\end{theorem}
\begin{proof}
	We have
	\begin{align*}
		|D_1\cup D_2|&<3\gamma+4\eta\gamma&\text{(by \cref{alem:size-D12})}\\
		|D_3|&\le 5(1+\e/5)\gamma_R\leq 5(1+\e/5)(1-\eta)\gamma	&\text{(by \cref{alem:size-D3})}\\
			\intertext{Thus,}
		|D_1 \cup D_2 \cup D_3| &< 8\gamma -5\eta\gamma+\epsilon\gamma-\epsilon\eta\gamma\\
		&\leq  (8+\epsilon) \gamma&\text{(maximized when $\eta=0$)}
	\end{align*}
%By
%\cref{alem:size-D12} we have $|D_1\cup D_2|<3\gamma+4\eta\gamma$.  Then,
%by \cref{alem:size-D3} we have $|D_3|\le 5(1+\e/5)\gamma_R\leq 5(1+\e/5)(1-\eta)\gamma$.
%%
%Therefore $|D_1 \cup D_2 \cup D_3| < 8\gamma -5\eta\gamma+\epsilon\gamma-\epsilon\eta\gamma$.
%%
%As $\eta\in[0,1]$, this is maximized when $\eta=0$. Hence
%\mbox{$|D_1 \cup D_2 \cup D_3|< (8+\epsilon) \gamma$}.
\end{proof}

The remaining theorems are proved analogously. 

\begin{theorem}\label{thm:bip}
  There exists a distributed LOCAL algorithm that, for every bipartite planar graph
  $G$, computes in a constant number of rounds a
  dominating set of size at most $(7+\epsilon)\gamma(G)$.
\end{theorem}
%\begin{proof}
%By
%\cref{alem:size-D12} we have $|D_1\cup D_2|<2\gamma+4\e\gamma$.  Then,
%by \cref{alem:size-D3} we have $|D_3|\le 10.5(1-\e)\gamma$.
%%
%Therefore $|D_1 \cup D_2 \cup D_3| < 12.5\gamma -6.5\e\gamma$.
%%
%As $\e\in[0,1]$, this is maximized when $\e=0$. Hence
%\mbox{$|D_1 \cup D_2 \cup D_3|< 12.5 \gamma$}.
%\end{proof}

\begin{theorem}\label{thm:girth}
  There exists a distributed LOCAL algorithm that, for every planar graph
  $G$ of girth at least~$5$, computes in a constant number of rounds a
  dominating set of size at most~$7\gamma(G)$.
\end{theorem}
%\begin{proof}
%By
%\cref{alem:size-D12} we have $|D_1\cup D_2|<3\gamma+\e\gamma$.  Then,
%by \cref{alem:size-D3} we have $|D_3|\le 4(1-\e)\gamma$.
%%
%Therefore $|D_1 \cup D_2 \cup D_3| < 7\gamma -3\e\gamma$.
%%
%As $\e\in[0,1]$, this is maximized when $\e=0$. Hence
%\mbox{$|D_1 \cup D_2 \cup D_3|< 7 \gamma$}.
%\end{proof}

\begin{theorem}\label{thm:outer}
  There exists a distributed LOCAL algorithm that, for every outerplanar graph~$G$, computes in a constant number of rounds a
  dominating set of size at most $(8+\epsilon)\gamma(G)$.
\end{theorem}

\section{Conclusion}

We simplified the presentation and generalized the algorithm of Czygrinow et al.~\cite{czygrinow2018distributed} from graph classes that exclude
some topological minor to graph classes~$\Cc$ where
$\nabla_1(G)$ is bounded
by an absolute constant for all $G\in \Cc$. This is a property
in particular possessed by classes with bounded expansion, which include
many commonly studied sparse graph classes. The obtained general  bounds are still large,
but by magnitudes smaller than those obtained in the original work of
Czygrinow et al.\cite{czygrinow2018distributed}.

It is an interesting and important question to identify the most
general graph classes on which certain algorithmic techniques work.
The key argument of \cref{lenzen-improved} works only for classes with~$\nabla_1(G)$ bounded by an absolute constant. We need different methods
to push towards classes with only $\nabla_0(G)$ bounded,
which are the degenerate classes.

 We then
provided a fine-tuned LOCAL algorithm that computes an
  ($11+\epsilon$)\hspace{1pt}-\hspace{1pt}approximation of
  a minimum dominating set in a planar graph in a constant number of
  rounds. Started with different parameters, the algorithm works also
  for several restricted cases of planar graphs. We showed that
  it computes an ($8+\epsilon$)\hspace{1pt}-\hspace{1pt}approximation for
  triangle-free planar graphs, a ($7+\epsilon$)\hspace{1pt}-\hspace{1pt}approximation
  for bipartite planar graphs, a 7\hspace{1pt}-\hspace{1pt}approximation
  for planar graphs of girth 5 and \linebreak an ($8+\epsilon$)\hspace{1pt}-\hspace{1pt}approximation
  for outerplanar graphs. In all cases except for the outerplanar case,
  where an optimal bound of 5 was already known, our algorithm
  improves on the previously best known approximation factors.
  This improvement is most significant in the case of general planar
  graphs, where the previously best known factor was 52.

\pagebreak
\bibliographystyle{abbrv}
\bibliography{ref}

%\appendix
%
%\input{13-cplex}

\end{document}